\documentclass[12pt, hidelinks]{article}
\usepackage[small, bf]{caption}
\usepackage[utf8]{inputenc}
\usepackage{fullpage}
\usepackage{appendix}
\usepackage{graphicx}
\usepackage{amsmath}
\usepackage{amssymb}
\usepackage{amsthm, amsfonts}
\usepackage{url}
\usepackage{mathtools}
\usepackage{chngcntr}
\usepackage{hyperref}
\usepackage{subdepth}

\usepackage{comment}

\newcommand{\ds}{{\Delta^{{\mathrm{sand}}}}}

\theoremstyle{definition}
\newtheorem{defn}{Definition}
\newtheorem{theorem}{Theorem}
\newtheorem{prop}{Proposition}
\newtheorem{lemma}{Lemma}
\counterwithin*{lemma}{subsection} 

\newtheorem{claim}{Claim}
\newtheorem{note}{Note}

\usepackage{tikz}
\usepackage{subcaption}
\usetikzlibrary{arrows,automata,calc}

\tikzset{
  treenode/.style = {align=center, inner sep=0pt, text centered,
    font=\sffamily},
  arn_n/.style = {treenode, circle, white, font=\sffamily\bfseries, draw=black,
    fill=black, text width=1.5em},
  arn_r/.style = {treenode, circle, red, draw=red, 
    text width=1.5em, very thick},
  arn_x/.style = {treenode, rectangle, draw=black,
    minimum width=0.5em, minimum height=0.5em}
}

\usepackage[style=alphabetic,backend=bibtex]{biblatex}
\bibliography{citations}

\title{Towards a Theory of Maximal Extractable Value I: Constant Function Market Makers}

\author{
Kshitij Kulkarni\\
{\small \texttt{ksk@eecs.berkeley.edu}}
\and 
Theo Diamandis\\
{\small \texttt{tdiamand@mit.edu}}
\and
Tarun Chitra\\
{\small \texttt{tarun@gauntlet.network}}
}
\date{April 2023}

\begin{document}

\newcommand{\ones}{\mathbf 1}
\newcommand{\reals}{{\mbox{\bf R}}}
\newcommand{\integers}{{\mbox{\bf Z}}}
\newcommand{\naturals}{{\mbox{\bf N}}}

\newcommand{\symm}{{\mbox{\bf S}}}  

\newcommand{\nullspace}{{\mathcal N}}
\newcommand{\range}{{\mathcal R}}
\newcommand{\Rank}{\mathop{\bf Rank}}
\newcommand{\Tr}{\mathop{\bf Tr}}
\newcommand{\diag}{\mathop{\bf diag}}
\newcommand{\card}{\mathop{\bf card}}
\newcommand{\rank}{\mathop{\bf rank}}
\newcommand{\conv}{\mathop{\bf conv}}
\newcommand{\prox}{\mathbf{prox}}

\newcommand{\Expect}{\mathop{\bf E{}}}
\newcommand{\Var}{\mathop{\bf Var{}}}
\newcommand{\de}{\Delta}

\newcommand{\Dom}{\mathop{\bf Dom}}
\newcommand{\Range}{\mathop{\bf Range}}
\newcommand{\Prob}{\mathop{\bf Prob}}
\newcommand{\Co}{{\mathop {\bf Co}}} 
\newcommand{\dist}{\mathop{\bf dist{}}}
\newcommand{\argmin}{\mathop{\rm argmin}}
\newcommand{\argmax}{\mathop{\rm argmax}}
\newcommand{\epi}{\mathop{\bf epi}} 
\newcommand{\Vol}{\mathop{\bf vol}}
\newcommand{\dom}{\mathop{\bf dom}} 
\newcommand{\intr}{\mathop{\bf int}}
\newcommand{\sign}{\mathop{\bf sign}}

\newcommand{\cf}{{\it cf.}}
\newcommand{\eg}{{\it e.g.}}
\newcommand{\ie}{{\it i.e.}}
\newcommand{\etc}{{\it etc.}}

\newcommand{\BEAS}{\begin{eqnarray*}}
\newcommand{\EEAS}{\end{eqnarray*}}
\newcommand{\BEA}{\begin{eqnarray}}
\newcommand{\EEA}{\end{eqnarray}}
\newcommand{\BEQ}{\begin{equation}}
\newcommand{\EEQ}{\end{equation}}
\newcommand{\BIT}{\begin{itemize}}
\newcommand{\EIT}{\end{itemize}}

\newcommand{\abs}[1]{\lvert{#1}\rvert}
\newcommand{\ip}[2]{\langle{#1}\,,\,{#2}\rangle}
\newcommand{\indicator}[1]{\mathbbm{1}_{\{{#1}\}}}
\newcommand{\graph}{{\mathcal G}}
\newcommand{\verts}{{\mathcal V}}
\newcommand{\edges}{{\mathcal E}}

\maketitle

\begin{abstract}
Maximal Extractable Value (MEV) refers to excess value captured by miners (or validators) from users in a cryptocurrency network. This excess value often comes from reordering users' transactions to maximize fees or from inserting new transactions that front-run users' transactions. One of the most common types of MEV involves a `sandwich attack' against a user trading on a constant function market maker (CFMM), which is a popular class of automated market maker. We analyze game theoretic properties of MEV in CFMMs that we call \textit{routing} and \textit{reordering} MEV. In the case of routing, we present examples where the existence of MEV both degrades and, counterintuitively, \emph{improves} the quality of routing. We construct an analogue of the price of anarchy for this setting and demonstrate that if the impact of a sandwich attack is localized in a suitable sense, then the price of anarchy is constant. In the case of reordering, we show conditions when the maximum price impact caused by the reordering of sandwich attacks in a sequence of trades, relative to the average price, impact is $O(\log n)$ in the number of user trades. Combined, our results suggest methods that both MEV searchers and CFMM designers can utilize for estimating costs and profits of MEV.\footnote{The code for all the numerical experiments in this paper can be found at this link: \url{https://github.com/tjdiamandis/mev-cfmm}}
\end{abstract}

\section{Introduction}
Public blockchains, including Bitcoin and Ethereum, allow any user to submit a transaction that modifies the shared state of the network.
Miners (or validators in proof-of-stake networks)\footnote{We use `miner' in this paper for consistency with the existing literature, \eg~\cite{daianFlashBoysFrontrunning2019}.} aggregate these transactions into blocks which they propose to the network. Each miner can propose blocks at a rate roughly proportional to the resources they have locked into the network.
Thus, the fixed fees each miner earns (\eg, block rewards in Bitcoin or staking yields in Ethereum) are also approximately proportional to these resources.
However, the transaction-dependent fees collected by the miner often vary dramatically from block to block. 

Consensus protocols have rules governing block validity, but the majority do not enforce constraints on transaction ordering \emph{within} a block. 
As a result, individual miners can propose blocks with a transaction ordering that nets them highest profit, possibly by inserting additional transactions into the block.
For example, a miner may observe a user's submitted decentralized exchange (DEX) trade and insert their own trade ahead of the user's trade to force this user to have a worse execution price. 
Any type of excess profit that a miner can extract by adjusting the execution of users' transactions is known as Maximal Extractable Value (MEV).\footnote{
    We note that attacks which incentivize unnecessary forking in blockchains, sometimes referred to as `time-bandit attacks', are also a form of MEV~\cite{judmayer2021sok}. 
    In this paper, we ignore these types of attacks but believe that our framework can be generalized to include them.
}
There are three principal agents involved in MEV: miners, network users, and MEV searchers.
Miners contribute resources to a network in order to win the chance to earn fees by validating transactions.
Network users are ordinary users who submit financial transactions to miners to be validated and added to the blockchain.
Finally, MEV searchers (or simply, `searchers') are agents who find profitable opportunities from reordering, inserting, or omitting transactions.

Searchers design \emph{strategies}: solutions to knapsack-like problems which find the most profitable sequence of transactions that fits within the block limit.
Competing searchers submit their proposed sequences of transactions to an auction to bid for inclusion in the next block offered by validators. The auction acts as a profit-sharing mechanism between the searchers, who look for opportunities, and the miners, who execute the proposed sequence of transactions.
We note that many expected the market to converge to an equilibrium in which miners and searchers are the same agent.
However, this convergence has not happened, partially due to services such as Flashbots~\cite{flashbots_mev_explore}, which runs a combinatorial auction allowing searchers who are not miners to bid on particular block positions.
These auctions have generated billions of dollars in excess revenue for validators since they were introduced\cite{flashbotsexplore}.

In this paper, we formalize a game theoretic view of MEV as it appears in decentralized exchanges that are implemented as constant function market makers (CFMMs), reviewed in~\S\ref{sec:mev_cfmms}. 
By viewing MEV as a multi-agent game between miners, searchers, and users, we can compare the equilibria that emerge from different forms of MEV.
This perspective allows us to analyze the economic properties of systems with MEV.

\paragraph{Prior work on MEV} Since MEV was first defined in 2019 ~\cite{daianFlashBoysFrontrunning2019}, miners and searchers have extracted over \$650 million~\cite{flashbots_mev_explore}. 
Moreover, the observed types of MEV strategies have grown rapidly~\cite{qin2020attacking, zhou2021a2mm, qin2021quantifying, angeris2021bundle, bartoletti2021maximizing, bartoletti2021theory}.
Thus, it is important to rigorously understand the space of possible MEV and quantify profitability under different conditions.
In \cite{bartoletti2021maximizing, heimbach2022eliminating}, the authors quantify sandwich attack profitability for constant product market makers, but they do not provide minimax, price of anarchy, or worst-case bounds for generic constant function market makers or for generic sequences of transactions.
Recent work has focused on reducing strategy profitability via more complex ordering consensus mechanisms (for example, approximate first-in-first-out sequencing~\cite{kelkar2020order} and other sequencing rules with execution price guarantees~\cite{ferreira2022credible}).
Fair sequencing forces validators come to consensus on relative transaction orderings, \ie, validators vote on whether transaction A came before transaction B as part of the consensus protocol rules.
Such systems, as the authors of~\cite{kelkar2020order} readily admit, cannot be deterministically secure due to the Condorcet paradox and Arrow's impossibility theorem. 
The same authors propose a permissioned blockchain solution~\cite{kelkar2021themis}, but the major blockchains in practice are not permissioned.
Moreover, these frameworks are application agnostic, yet, in practice, the profitability of MEV strategies often varies dramatically from one application to another~\cite{flashbots_mev_explore}.
Subsequent work~\cite{babel2021clockwork} proposed a framework for numerically evaluating profitability so that results are application-specific.
This framework uses formal verification to search through the set of sequences of actions that satisfy some predicate (\eg,~``strategy generates more than \$X of profits'').
While this work is practically useful, it provides no theoretical insight into why certain applications have more or less MEV than others.
Moreover, no prior work provides any guidance on the types of economic equilibria that can occur in systems with MEV.

\paragraph{MEV in CFMMs.}
Constant function market makers (CFMMs) are decentralized exchanges that have seen widespread use in blockchains \cite{angeris2021cfmm}. Sandwich attacks refer to a form of MEV in which a searcher forces a user of an exchange to have a worse execution price (possibly up to their `slippage limit') and then profits from this artificial price movement.
Specifically, a searcher can `sandwich' or pad the user's transaction with trades before and after to buy low, force the user to have a worse execution price, and then sell high.
Sandwich attacks are by far the most popular type of MEV, with over \$500m extracted from users via sandwich attacks~\cite{flashbots_mev_explore}.
We demonstrate that, provided there is enough liquidity, the maximal profit attainable from a sandwich attack has a particular subadditivity property, akin to the subadditivity for privacy found in~\cite{chitra2021differential}.
This explicit calculation shows that applying the privacy methodology of~\cite{chitra2021differential} reduces sandwich attacks and explicitly shows that MEV and privacy in CFMMs are inversely related (\ie,~lower MEV leads to higher statistical privacy and vice-versa).

Prior work has focused on analyzing sandwich attacks only in one type of CFMM, the constant product market maker Uniswap~\cite{heimbach2022eliminating, zust2021analyzing}.
In \S\ref{sec:mev_cfmms}, we generalize this analysis to \emph{any} CFMM as defined in~\cite{angeris2021cfmm}.
To do so, we first define a sandwich attack in terms of the forward exchange function of a CFMM.
This definition allows us to utilize the notion of CFMM curvature~\cite{angerisCurvature} to explicitly bound the profitability of a sandwich attack, given a particular user trade.
Unlike prior work~\cite{heimbach2022eliminating, bartoletti2021maximizing}, we also consider MEV profitability for a searcher that can sandwich sequence of trades $\Delta_1, \ldots, \Delta_n$ rather than a single trade $\Delta$. 
Finally, we note that MEV in CFMMs has implicitly been analyzed when studying privacy preserving mechanisms in CFMMs, including threshold cryptography~\cite{agrawal_osmosis, ferveo}, differential privacy~\cite{chitra2021differential}, and zero knowledge commitments~\cite{zswap}.

\paragraph{Main Results.}
In this work, we aim to answer two questions regarding MEV in CFMMs:
\begin{itemize}
    \item In the case of a network of CFMMs trading multiple assets, how much does the presence of sandwich attackers on the network affect the \emph{routing} of trades?
    \item In the case of a single CFMM trading two assets, how much does \emph{reordering} a sequence of user trades affect the excess price impact caused by sandwich attacks?
\end{itemize}
We answer the first question in~\S\ref{sec:trade_flow_routing} by adapting conventional price of anarchy (PoA) results~\cite{roughgarden2005selfish, roughgarden2015intrinsic, roughgarden2017price} to CFMMs. We consider the routing of a single aggregate trade across a network of CFMMs. We define \emph{selfish routing}, in which trades on each path in the network try to get the maximum pro-rata share of the output from that path. This results in the notion of an equilibrium splitting of a trade.
We compare selfish routing to optimal routing~\cite{angeris2022optimal}, which seeks to maximize the net output from the network. We then analyze the gap between selfish and optimal routing, also known as the price of anarchy, when there are sandwich attackers on the network. The presence of sandwich attacks shifts both optimal and selfish routing.
Our main result shows that the price of anarchy is bounded by a constant for any sized sandwich attack, which we establish using the $(\lambda, \mu)$-smoothness results of~\cite{roughgarden2015intrinsic}
In addition, we construct a CFMM network that, perhaps counterintuitively, avoids Braess paradox-like~\cite{roughgarden2005selfish} behavior after a sandwich attacker is introduced, as this attacker makes the Braess edge more expensive. This example suggests that sandwich attackers can sometimes \emph{improve} the quality of selfish routing. 

We answer the second question in~\S\ref{sec:reordering} by constructing an analogue of a competitive ratio or so-called `prophet inequality'~\cite{hill1983prophet} that measures the ratio of the price impact caused by the worst case sandwich attack to that caused by the average sandwich attack, given an order flow of $n$ trades $\Delta_1, \ldots, \Delta_n$ across a single CFMM.
We call this ratio the `cost of feudalism'.
We show that, under sufficient liquidity conditions, the cost of feudalism is $O(\log n)$.
This result suggests that there is not a large asymptotic difference between the worst sandwich attack and the average sandwich attack, provided that the CFMMs involved have sufficient liquidity.
We note that our liquidity constraints are a measurement of `locality' of a sandwich attack, which ensures that the compounding price impact of a sequence of sandwiches is bounded sufficiently.
We also note that when this locality condition is not satisfied, the general version of a sandwich attack is a potentially hard knapsack problem, where the sandwich attacker must decide which subset of $n$ trades he would like to jointly sandwich. This locality result generalizes the results of~\cite{bartoletti2021maximizing}, which only considers optimal sandwiches for Uniswap in the arbitrarily large block size limit.

\section{Sandwich Attacks}\label{sec:mev_cfmms}
In sandwich attacks ~\cite{qin2020attacking}, an adversary, called a sandwich attacker, places orders before and after a user's order to force the user's order to have a worse execution price.
When placing an order, users specify both trade side and a limit price (in the form of a slippage limit).
The slippage limit prevents the order from being executed at a price that is much worse than the current market price.
After seeing the user's order, an adversary can submit a trade before the user's trade, pushing up the user's execution price to any value below the slippage limit.
Then, immediately after the user's trade is executed, the adversary places a trade in the opposite direction to recover their initial investment and a profit resulting from the price impact of user's trade.
We will first describe sandwich attacks concretely for the most widely-used CFMM, Uniswap, before describing their characteristics for general CFMMs.

\subsection{Uniswap}\label{sec:uniswap}
Before considering generic CFMMs, we will first illustrate our results for Uniswap~\cite{angerisImprovedPriceOracles2020, angeris2019analysis, uniswap}.
Uniswap was the first CFMM to launch in production and has had over 1 trillion dollars of trading volume flow through it since inception~\cite{dexDune}.
It has a particularly simple structure: assume we have reserves of token A, $R$, and reserves of token $B$, $R'$, and without loss of generality, that token $B$ is the numeraire. Then, a user's trade of size $\Delta$ is valid if:
\begin{equation}\label{eq:uni_invariant}
(R  -  \Delta' )(R' + \gamma  \Delta) = R R'
\end{equation}
where $1-\gamma$ represents the percentage fee parameter that controls how much the liquidity provider charges for facilitating the trade. This trade can be thought of as the user providing $\Delta$ units of token $B$ for some amount of token $A$, or analogously, the CFMM or liquidity provider providing some amount of token $A$ for $\Delta$ units of token $B$.
The amount of token $A$ $\Delta'$ is determined implicitly by~\eqref{eq:uni_invariant} and will vary as the fee $\gamma$ is changed.
The \emph{forward exchange rate} quoted by Uniswap is the ratio of the reserves, \ie~the price of A in terms of B is $p_{AB} = \frac{R'}{R}$.
For a trade of size $\Delta$, we will define $p_{AB}(\Delta, R, R') = \frac{R' + \Delta'}{R -   \gamma \Delta'}$. This is the \emph{marginal forward exchange rate} of trade $\Delta$. That is, by changing the reserves, the trade changes the price of the tokens. The amount of output token that the user receives can then be computed using the forward exchange rate as  $G(\Delta) = \frac{1}{\gamma} \left( \frac{-R R'}{R + \Delta} + R'\right)$ \cite{angeris2019analysis}.
Note that when the fee $\gamma = 1$, the quantity $k = R R'$ is always constant.
In the rest of this paper, we will always work in the feeless regime $(\gamma = 1)$, but we note that lower bounds from~\cite[App. B]{angerisCurvature} can be used to generalize the results of this paper to the case when $\gamma < 1$.

\paragraph{Slippage Limits.} 
A \emph{slippage limit} $\eta \in [0,1]$ represents how much of a price impact a user is willing to tolerate to execute their trade, as measured by the minimum amount of the output token they are willing to receive. 
The user has to provide the slippage limit because, in general, they may not know the value of the reserves immediately before their trade is executed.
For instance, suppose there are two trades $\Delta_1, \Delta_2$ to be executed by two different users.
Since miners get to choose whether they execute $\Delta_1$ first or second, the user does not know if their trade is executed with an initial price $p_{AB}(0, R, R')$ or at $p_{AB}(\Delta_2, R, R')$.
The slippage limit $\eta$ is a way for a user to say that they do not want to receive less than $1-\eta$ times the nominal amount they would receive in the absence of the other trade, \eg,~if the trade is executed in the sequence $\de_2, \de_1$, the miner cannot execute the trade unless $G(\de_1 + \de_2) - G(\de_2) \geq (1-\eta)G(\Delta_1)$.
This slippage limit is enforced by the Uniswap smart contract and allows users to ensure that their trade is executed at a favorable price.\footnote{The condition enforced by the Uniswap contract is a \textit{minimum amount} of the output token that a user is willing to accept \cite{uniswapSlippage}. 
}
However, this feature still places the onus of choosing the correct parameter $\eta$ on the user.
We define a \emph{trade} to be a pair $(\Delta, \eta)$ of a trade size and slippage limit. A user's slippage limit is loose if there exists a trade $\de'$  executed before $(\Delta, \eta)$ such that $G(\de + \de') - G(\de') > (1+\eta)G(\Delta)$, \ie, a trade can be inserted before the user's trade.
This situation allows for a sandwich attacker to construct a trade $\ds$ such that $G(\de+ \ds) - G(\ds) = (1+\eta)G(\Delta)$.
By filling up the slack in the inequality constraint, the attacker worsens the execution price of the user's trade $(\Delta, \eta)$.
Moreover, if the attacker submits a trade of size $\ds'$ after executing the trades $\ds$ and $\Delta$, then they are able to profit due to the convexity of the Uniswap invariant (see, e.g.,~\cite{angerisCurvature} and the example below). We will now give a concrete example of a sandwich attack on Uniswap before defining sandwich attacks generally. 

\paragraph{Example of Sandwich Attack on Uniswap.}
We illustrate a concrete example of a sandwich attack in the case of Uniswap, whose forward exchange function takes the form:
\begin{align*}
    G(\Delta)= -\frac{k}{R + \Delta} + R'
\end{align*}
for reserves $R$ and $R'$ of input and output asset, respectively, and $k = R R'$ \cite{angeris2019analysis}. Assume the user submits a trade $(\de, \eta)$ to Uniswap, and a sandwich attacker wants to design $\ds$ to force the slippage limit to be tight. That is, $\ds$ satisfies:
\begin{align*}
    G(\de + \ds) - G(\ds) = (1-\eta) G(\de).
\end{align*}
Plugging in the functional form of $G(\cdot)$ for Uniswap, we have:
\begin{align*}
    -\frac{k}{R +\de + \ds} + R' + \frac{k}{R + \ds} - R' = (1-\eta) \left( -\frac{k}{R +\Delta} + R'\right).
\end{align*}
Finally, solving for $\ds$ (with the full calculation in Appendix \ref{sec:uniswapAppendix}), we find the optimal sandwich attack $\ds$: 
\begin{align}\label{eq:uniswapSandwich}
     \ds =  \frac{-(\de + 2R) + \sqrt{( \de+2R)^2 -4 (R^2 + R \de) \frac{-\eta}{1-\eta}}}{2} 
\end{align}
Note that if $\eta = 0$, then $\ds =0$, as expected. In addition, we see that $\ds$ is an increasing function of $\eta$. This demonstrates that as the user is willing to tolerate a smaller minimum output, the amount the sandwich attacker can use to fill the slack in the user's trade increases.

\subsection{Constant function market makers}\label{s-cfmm}
We now generalize sandwich attacks to \emph{constant function market makers} (CFMMs).
CFMMs hold some amount of \emph{reserves} $R, R' \ge 0$ of two assets
and have a \emph{trading function} $\psi: \reals^2\times\reals^2 \to \reals$.
Users can then submit a \emph{trade} $(\Delta, \Delta')$ denoting the amount they wish to tender (if negative) or receive (if positive) from the market.
The contract then accepts the trade if $\psi(R, R', \Delta, \Delta') = \psi(R, R', 0, 0)$,
and pays out $(\Delta, \Delta')$ to the user.

\paragraph{Curvature.} 
We briefly summarize the main definitions and results of~\cite{angerisCurvature} here.
Suppose that the trading function $\psi$ is differentiable (as most trading functions in practice are), then
the \emph{forward exchange rate} for a trade of size $\Delta$ is
$
g(\Delta) = \frac{\partial_3 \psi(R, R', \Delta, \Delta')}{\partial_4 \psi(R, R', \Delta, \Delta')}.
$
Here $\partial_i$ denotes the partial derivative with respect to the $i$th argument,
and $\Delta'$ is specified by the implicit condition $\psi(R, R', \Delta, \Delta') = \psi(R, R', 0, 0)$; \ie, the trade
$(\Delta, \Delta')$ is assumed to be valid. Additionally, the reserves $R, R'$ are assumed to be fixed.
Matching the notation of Section~\ref{sec:uniswap}, the function $g$ represents the marginal forward exchange rate of a positive-sized trade.
We say that a CFMM is \emph{$\alpha$-stable} if it satisfies 
\begin{align*}
    g(0) - g(-\Delta) \leq \alpha \Delta
\end{align*}
for all $\Delta \in [0, M]$ for some positive $M$. This condition provides a linear upper bound on the maximum price impact that a trade bounded by $M$ can have.
Similarly, we say that a CFMM is $\beta$-liquid if it satisfies
\begin{align*}
    g(0) - g(-\Delta) \geq \beta \Delta
\end{align*}
for all $\Delta \in [0, K]$ for some positive $K$. One important property of $g$ is that it can be used to compute $\Delta'$~\cite[\S2.1]{angerisCurvature}:
\begin{equation}\label{eq:curvOut}
    \Delta' = \int_{0}^{-\Delta} g(t) dt.
\end{equation}
Simple methods for computing $\alpha$ and $\beta$ in common CFMMs are presented
in~\cite[\S1.1]{angerisCurvature} and~\cite[\S4]{angeris2021cfmm}. We define $\de' = G(\Delta)$ to be the \emph{forward exchange function}, which is the amount of output token received for an input of size $\Delta$. Whenever we reference the function $G(\Delta)$ for a given CFMM, we always make clear the reserves associated with that CFMM. We note that $G(\Delta)$ was shown to be concave and increasing in \cite{angeris2021cfmm}. 

\paragraph{Two-sided bounds.}\label{sec: twoSidedBounds}
We can define similar upper and lower bounds for $g(\Delta) - g(0)$, with constants $\mu'$ and $\kappa'$, which hold when the trades $\Delta$ are in intervals $[0, M'], [0, K']$, respectively.
For the remainder of this paper, we will refer to $\alpha$-stability as the upper bound for both $g(0) - g(-\Delta)$ and $g(\Delta) - g(0)$, and similarly refer to $\beta$-liquidity as double-sided lower bounds.
More specifically, given $\mu, \mu'$, we say that a CFMM is symmetrically $\alpha''$-stable if $|g(\Delta) - g(0)| \le \alpha'' |\Delta|$, when $-M \le \Delta \le M'$, and symmetrically $\beta''$-liquid if $|g(\Delta) - g(0)| \ge \beta'' |\Delta|$
when $-K \le \Delta \le K'$. From the above, it suffices to pick $\alpha'' = \min\{\alpha, \alpha'\}$ and $\beta'' = \min\{\beta, \beta'\}$. 

Note that any two-sided $\alpha$-stable and $\beta$-liquid market maker is automatically $\min(\alpha, \beta)$-stable and $\min(\alpha,\beta)$-liquid.
An $\eta$-liquid and $\eta$-stable forward exchange rate function is `bi-Lipschitz' and admits an inverse $g^{-1}(p)$ that is also bi-Lipschitz~\cite{howard1997inverse}.
In particular, if $g$ is $\eta$ bi-Lipschitz, then $g^{-1}(p)$ is $\frac{1}{\eta}$ bi-Lipschitz, \ie,~$\frac{1}{\eta} p \leq |g^{-1}(p) - g^{-1}(0)| \leq \frac{1}{\eta} p$. 

\paragraph{Slippage Limits.}
Analogously to the case of Uniswap, when a user submits an order to a CFMM, they submit two parameters: a trade size $\Delta \in \reals$ and a \emph{slippage} $\eta \in [0, 1]$.
The slippage is interpreted as the minimum output amount that the user is willing to accept as a fraction of $G(\Delta)$. That is, the trade is accepted if the amount in output token the user receives is larger than or equal to $(1-\eta) G(\Delta)$.

\subsection{Sandwich Attacks}\label{sec:sandwich}
We generalize prior work analyzing sandwich attacks to CFMMs with two-sided bounds on their price impact functions $g(\Delta)$.
Recall that a user submits a trade to a CFMM of the form $T = (\Delta, \eta) \in \reals \times [0,1]$, where $\eta$ is the slippage limit. If a user submits an order that is not tight, then there exists a $\Delta^{\mathrm{sand}}$ such that $G(\Delta + \Delta^{\mathrm{sand}}) - G(\ds ) > (1-\eta) G(\Delta)$. That is, $\ds$ satisfies:
\begin{align}\label{eq:sandwich-defn}
    G(\Delta + \ds ) - G(\ds) = (1-\eta) G(\Delta).
\end{align}
One can use the equation $G(\Delta + \ds ) - G(\ds) = (1-\eta) G(\Delta)$ to numerically solve for the optimal $\ds$ by finding the roots of $G(\Delta + x) - G(x) - (1-\eta) G(\de) = 0$. 
\begin{center}
\begin{table}[htb]
    \centering
    \begin{tabular}{ |c| c| c| }
    \hline 
      & \text{Input Reserves} & \text{Output Reserves} \\ 
    \hline
     Sandwich attack & $R \rightarrow R+\ds$ & $R' \rightarrow R' - \Delta^{\mathsf{sand, out}}$ \\
     \hline
     User submits trade & $R \rightarrow R +\ds + \de$ & $R' \rightarrow R' - \Delta^{\mathsf{sand, out}} - \Delta^{\mathsf{out}}$ \\ 
     \hline
     sandwich attacker sells back & $R \rightarrow R + \ds + \de - \ds'$ & $R \rightarrow R  -\Delta^{\mathsf{out}}$ \\ \hline
    \end{tabular}
    \caption{Sequence of reserve updates in a sandwich attack.}
    \label{Tab:sandwichTable}
\end{table}
\end{center}
\noindent To see where equation \eqref{eq:sandwich-defn} comes from, we enumerate the trade sequence of a sandwich attack in Table \ref{Tab:sandwichTable}. Suppose initially that the CFMM has reserves $R$ and $R'$. The sandwich attacker submits $\ds$ ahead of the user, which causes the reserves to be updated as $R \rightarrow R+\ds$ and $R' \rightarrow R' - \Delta^{\mathsf{sand, out}}$, where $\Delta^{\mathsf{sand, out}}$ is implicitly given by the trading function, that is, $\Delta^{\mathsf{sand, out}} = G(\ds)$. Next, the user submits the trade $\de$, after which the reserves are $R \rightarrow R + \ds + \de$ and $R \rightarrow R - \Delta^{\mathsf{sand, out}} - \de^{\mathsf{out}}$.  Recall that $\ds$ is constructed so that the user receives no less than $(1-\eta) G(\Delta)$ units of the output token. The amount the user receives after sandwiching, $\de^{\mathsf{out}}$, is given by $\Delta^{\mathsf{sand, out}} + \de^{\mathsf{out}} = G(\ds + \de)$. Substituting for $\Delta^{\mathsf{sand, out}}$, we have that $\Delta^{\mathsf{out}} = G(\ds + \de) - G(\ds)$. 
\\ \\ 
\noindent We assume that the sandwich attack is constructed optimally, so $\Delta^{\mathsf{out}}$ is equal to the minimum amount the user is willing to receive, $(1-\eta) G(\Delta)$, (\cf~\eqref{eq:sandwich-defn}). Note that we have abused notation by not explicitly denoting the reserves at the various stages of the sandwich attack in the function $G(\cdot)$. The reserves at each step are explicitly written in Table~\ref{Tab:sandwichTable} and must be taken into account when applying a forward exchange function $G(\cdot)$.
\\ \\ 
\noindent After $\ds$ and $T$ are executed, the sandwich attacker sends a trade of $\ds'$ to recover their initial investment of $\ds$ and make a profit. After sending the initial trade of $\ds$, the sandwich attacker holds $G(\ds)$ of output token. The attacked thus sells back the amount of output token they hold, $G(\ds)$, which defines $\ds'$ as follows, in units of input token:
\begin{align}\label{eq: inversesandwich-defn}
    \ds' = \ds + \de - G^{-1}(G(\de + \ds) - G(\ds)),
\end{align}
where $G^{-1}(\cdot)$ is the reverse exchange function,\ie, the inverse of $G$ \cite{angeris2021cfmm}.
\\ \\
\noindent Therefore, we can define a \emph{sandwich attack} as a triplet of transactions: $\ds(\Delta, \eta), (\Delta, \eta), \ds'(\Delta, \eta)$. We emphasize that both $\ds(\Delta, \eta)$ and $\ds'(\Delta, \eta)$ are in units of input token, are functions of $\Delta$ and $\eta$, and solve the equations \eqref{eq:sandwich-defn} and \eqref{eq: inversesandwich-defn}. If the sandwich attack is executed, the sandwich attacker can make a profit of: 
\begin{align}\label{eq:profit-defn}
    \mathsf{PNL}(\Delta, \eta) & =  \ds'(\de, \eta)  - \ds(\de , \eta) = \de - G^{-1}(G(\de + \ds) - G(\ds))
\end{align}
measured in input token, where $\ds (\Delta, \eta)$ refers to the solution of the implicit equation \eqref{eq:sandwich-defn} and $\ds'(\de, \eta)$ refers to the quantity defined in \eqref{eq: inversesandwich-defn}. Note that when $\eta = 0$, $\mathsf{PNL}(\Delta, \eta) = 0$ for all $\Delta$, as desired. Frequently, we will abuse notation by dropping the dependence of the sandwich attack on the user trade $\Delta$ and the slippage limit $\eta$, and just denote the sandwich attack by $\ds$ and $\ds'$.

\subsection{Bounds on Sandwich Attack Profitability.}
In order to reason about the impact of sandwich attacks, we first need to determine the expected size of a sandwich attack given a sequences of trades $\Delta_1, \ldots, \Delta_n$. We first show upper and lower bounds on the sandwich trade size as a function of curvature parameters and slippage limits. We will show that sandwich attack profitability is often maximized by sandwiching each trade $\Delta_i$ independently.
This result underscores the `locality' of sandwich attacking---one doesn't need to combine subsets of trades to sandwich together. This locality reduces computational complexity for searchers and allows us to bound the net price impact of sandwich trades.
To construct these bounds, we first need to define the rate of growth of $G(\Delta)$:
\begin{defn}
A forward exchange function $G(\Delta)$ is \emph{$(\mu, \kappa)$-smooth} if there exists $M > 0$ such that for all $\Delta \in [0, M]$ there exist constants $\mu, \kappa > 0$ such that 
\begin{align}\label{eq:capG-bd}
   \kappa \Delta \leq  G(\Delta) - G(0) \leq \mu \Delta 
\end{align}
\end{defn}
Usually $G(0) = 0$, so these inequalities correspond to a set of bilipschitz bounds on $G$.
We can define an analogous notion of smoothness for the reverse exchange function (and bounds for $\Delta < 0$), but, for simplicity, we will phrase all of our results in terms of the forward exchange functions. All of the proofs also hold for reverse exchange functions.
\begin{note}
The constants $\mu$ and $\kappa$ in \eqref{eq:capG-bd} are \textit{distinct} those in the definitions of $\alpha$-stability and $\beta$-liquidity in Section~\ref{sec: twoSidedBounds}.
\end{note}

\paragraph{Bounds on $\ds$.}
Using the smoothness constants, we can bound $\ds$ (proofs are in Appendix~\ref{app:dsbounds}).
\begin{claim}\label{ds_ub_claim}
If $\eta \geq 1 - \frac{\kappa}{\mu}$ then we have $\ds(\eta, \Delta) = O(\eta)\Delta$.
\end{claim}
\noindent This bound demonstrates that the size of a sandwich attack is linear in the slippage limit, provided that the slippage limit is sufficiently larger than a curvature ratio.
Such a bound can be used, for instance, by wallet designers to help users choose slippage limits that explicitly bound the maximum expected sandwich attacker profit.
For lower bounds on $\ds$, we will need to make further assumptions.
In particular, we require that the price impact function $g(\Delta) = G'(\Delta)$ grows sufficiently fast.
\begin{claim}\label{ds_lb_claim}
Suppose that the forward exchange rate $g(\Delta)$ is $\beta$-liquid in addition to $G$ being $(\mu, \kappa)$-smooth. Then there exists $\zeta = 1 + \Theta(\sqrt{1+\eta})$ such that $\ds \geq \left(\frac{\mu}{\beta} - \Delta\right)\zeta$.
\end{claim}
\noindent We gain intuition for why we need this extra assumption by analyzing a constant sum market maker which has $\beta = 0$ \cite{angeris2021cfmm}.
In a constant sum market maker, there is no sandwich profit, as there is no price impact when one executes the sequence of trades $(\ds, \Delta, \ds')$.
Therefore, in order for us to lower bound the sandwich attack size (and profit, which is linear in $\ds$ as per \ref{eq:profit-defn}), we need some non-zero price impact.
This non-zero price impact is specified by the $\beta$-liquid condition on $g$.

\paragraph{Bounds on $\ds'$.}
Now, we bound the round trip trade made by the sandwich attacker, $\ds'$, which satisfies equation \eqref{eq: inversesandwich-defn} (note that $\ds'$ is in units of input token).
We prove the following two claims in Appendix \ref{app:dsprime}.
\begin{claim}\label{claim:dsprime-ub}
Suppose that $\eta \geq \frac{\mu-\kappa}{\mu}$. Then $\ds' = O(\eta) \Delta$.
\end{claim}
\noindent This claim demonstrates that under mild conditions on the slippage limit, we can control the roundtrip profit in terms of linear factors of the slippage limit. Note that in the next bound, we require that $g(\cdot)$ is $\beta$-liquid for the same reason as above: there needs to be some excess price impact that the sandwicher can cause for the sandwich to be profitable.
\begin{claim}\label{claim:dsprime-lb}
Suppose that $g(\Delta)$ is $\beta$-liquid. Then there exists $\gamma = 1 + \Theta(\sqrt{1+\eta})$ such that
$\ds' \geq \frac{\mu \gamma}{\beta} - \Delta\left(\gamma+\frac{\eta\mu}{\kappa}\right)$.
\end{claim}

\paragraph{Upper Bound on Sandwich Attacker Profit.}
Recall that the profit of a attack is controlled by $\ds(\Delta,\eta) - \ds'(\Delta, \eta)$.
To upper bound price impact, we first need to lower bound $\ds'$.
Using Claims \ref{ds_ub_claim} and \ref{claim:dsprime-lb}, we can bound the total trade size that occurs in the input token when sandwiched (the effective size of the sandwich attack): $\ds + \Delta - \ds' \leq (O(\eta)+\gamma)\Delta - \frac{\mu\gamma}{\beta}$.
Similarly, this gives us a bound on the profit \eqref{eq:profit-defn}:
\begin{align*}
    \mathsf{PNL}(\Delta, \eta) = \ds' - \ds \leq  \left(O(\eta) + \gamma - 1\right)\Delta - \frac{\mu\gamma}{\beta} \leq C \max(\eta, \sqrt{1+\eta}) \Delta - \frac{\mu\gamma}{\beta}.
\end{align*}
This bound implies that, given all of the liquidity and slippage conditions of the claims are met, profit and price impact are linear in $\eta$.
Using the precise constants in Appendices \ref{app:dsbounds} and \ref{app:dsprime}, one can also compute a `hurdle rate' in terms of $\gamma$ which describes minimal conditions for a sandwich attacker to be profitable (proof in Appendix \ref{app:dsprime}):
\begin{claim}\label{claim:hurdle_rate}
If $\left(\eta\left((1+\frac{\mu}{\kappa}\right) - \left(2-\frac{\kappa}{\mu}\right) + \gamma\right)\Delta \geq \frac{\mu\gamma}{\beta}$ then $\mathsf{PNL}(\Delta, \eta) \geq 0$.
\end{claim}
\noindent This simple result can be used by both wallet designers (who are optimizing $\eta$ for users) and protocol designers (who can control $\mu$ and $\kappa$) as a way to minimize expected sandwich profit.

\paragraph{Sandwich Profitability is Local.}
The sandwich attacker's net profit in units of input token is $\ds' - \ds$. That is, they put in $\ds$ input tokens in and receive $\ds'$ input tokens out.
When we reason about two trades $(\Delta_1, \eta_1), (\Delta_2, \eta_2)$, we assume that we are dealing with a single slippage limit $\eta = \min(\eta_1, \eta_2)$. That is, we only account for slippage that the most conservative user inputs. Next, we introduce the following definition of locality of sandwich attacks:
\begin{defn}\label{stronglyLocal}
We say that sandwich attacks for a sequence of trades $T = \{(\Delta_1, \eta_1), \ldots, (\Delta_n,\eta_n)\}$ are \emph{strongly local} if we have for any index set $\{i_1, \dots, i_J\}$, with $i_1 < \dots < i_J$, $i_1 \geq 1$, $i_J = n$, and $J \leq n$, we have: 
$\mathsf{PNL}(\de_1 + \dots + \de_{i_1}) + \dots + \mathsf{PNL}(\de_{i_{j-1} +1} + \dots+ \de_{i_J})  \leq \sum_{i=1}^{n} \mathsf{PNL}(\de_i)$
where $\mathsf{PNL}(\de_1 +\dots+  \de_j) = \ds'(\de_1+ \dots + \de_j, \eta) - \ds(\de_1+ \dots + \de_j, \eta)$.
\end{defn}
\noindent Informally, this definition says that it is never more profitable to sandwich bundles of transactions instead of simply sandwiching them individually.
This definition also implicitly constrains the curvature constants, insofar as they cannot be too large (\ie, situations with very low liquidity or very high price impact).
When a block contains at most $3k$ trades, the absence of such locality in the sandwich attacks implies that the sandwich attacker has to solve a knapsack problem to pick a partition $P$ of $[k]$ to sandwich attack. This knapsack problem can be intractable when $k$ is large.
In Appendix~\ref{app:sandwichLocality}, we prove the following statement showing conditions under which it is never optimal to sandwich pairs of transactions, which we call \textit{pairwise locality}. 
\begin{prop}
Suppose that we have trades $T = \{(\Delta_1, \eta), \ldots, (\Delta_n, \eta)\}$ passing through a CFMM with curvature constants $\mu, \kappa$. Then, given \textit{sufficient conditions} in Equation \eqref{eq: sufficientConditionLocality}, sandwich attacks are \textit{pairwise local}. That is for all $i \in [n]$:
$\mathsf{PNL}(\de_i + \de_{i+1}) \leq \mathsf{PNL}(\de_i) + \mathsf{PNL}(\de_{i+1})$
\end{prop}
The interpretation of this sufficient condition is that there must be sufficient liquidity such that the cumulative price impact of sandwiching two adjacent trades is not larger than the price impact of an individual trade.
\noindent In the following, we always maintain the assumption that sandwich attacks are strongly local, following Definition \ref{stronglyLocal}.

\section{Routing MEV}\label{sec:trade_flow_routing}
Next, we introduce the notion of routing MEV for sandwich attacks. Intuitively, routing MEV is the excess value that a sandwich attacker can extract from a user trades over a network of CFMMs.
Given this network of CFMMs, some amount of an input token, and a desired output token, we call a sequence of trades that converts all of the input token to some amount of the output token a \emph{route} through the network.
We can efficiently construct a route that maximizes the output token amount because the \emph{optimal routing} problem (without transactions fees) is a convex optimization problem~\cite{angeris2022optimal}.
\\ \\
\noindent 
We consider the setting where multiple users wish to trade token $A$ for token $B$, and trades have been aggregated into a single trade to be routed across a network of CFMMs. 
We assume that the output of each path is distributed pro-rata to all users who trade through that path.
To measure the impact of sandwich attacks on this aggregate trade, we define \emph{selfish routing} as the scenario in which users selfishly route their component of the trade to optimize their own output. 
We will show that selfish routing leads to congestion, \ie, worse prices for users that choose to take routes that other users are also taking.
In the case of selfish routing, an equilibrium is an allocation of order flow to routes such that the average price of the output token among all used paths is equal.
(We will make this precise shortly.)
Intuitively, this equilibrium condition means that a user cannot gain more output token by switching paths.
We will compare selfish routing to \emph{optimal routing}, where a `central planner' can allocate the order flow across paths to maximize the amount of output.
\\ \\
We define the \emph{price of anarchy} as the ratio of the output under optimal routing to the output under selfish routing, possibly in the presence of sandwich attackers on the CFMM network.
The main result of this section is to prove that the price of anarchy is constant and bounded by constants related to both the slippage limits defined by the user and the liquidity of the CFMM network.
Before introducing the formal definitions of the above quantities, we provide two illustrative examples that show cases in which sandwich attacks worsen optimal routing, but, counterintuitively, improve selfish routing on CFMM networks. The code for all the numerical results in this section can be found at this link: \url{https://github.com/tjdiamandis/mev-cfmm}.

\subsection{The CFMM Pigou Example}\label{sec:pigou}
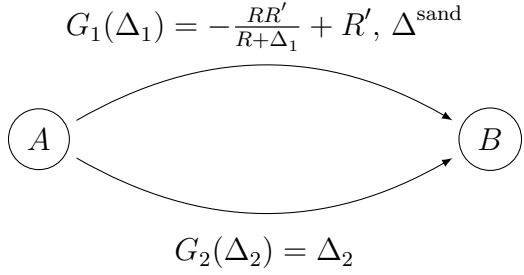
\begin{figure}[h]
\centering
\begin{tikzpicture}
\node[draw, circle] at (-.5,0) {$A$};
\node[draw, circle] at (5.5,0) {$B$};
\draw[->, >=latex] (0, 0.25) to[out=30, in=150] (5,0.25);
\node at (2.5, 1.5) {$G_1(\Delta_1) = -\frac{RR'}{R+\Delta_1} + R'$, $\ds$};
\draw[->, >=latex] (0, -0.25) to[out=-30, in=-150] (5,-0.25);
\node at (2.5, -1.5) {$G_2(\Delta_2) = \Delta_2$};
\end{tikzpicture}
\caption{The CFMM Pigou Network, for $R = 1, R' = 2$, and $\Delta = 1$.}
\label{fig:pigou}
\end{figure}

\paragraph{Sandwich attackers impact routing.}
Here we provide an explicit example of how optimal and selfish routing change when there is a sandwich attacker on the Pigou network shown in Figure~\ref{fig:pigou}.
Users desire to trade between tokens A and B and have two CFMMs on which to trade.
Suppose that CFMM $1$ is a constant product CFMM (\eg, Uniswap) with reserves $(R, R')$, and CFMM $2$ is a constant sum CFMM which always quotes the same forward exchange rate (until its reserves are depleted).
We denote the forward exchange functions of these CFMMs by $G_1$ and $G_2$ respectively.
Recall from \S\ref{sec:uniswap} that Uniswap has a forward exchange function of the form $G_1(\Delta) = -\frac{RR'}{R+\Delta} + R'$,
and the forward exchange rate is $g_1(\Delta) = \frac{RR'}{(R + \Delta)^2}$.
The forward exchange function for the constant sum CFMM is $G_2(\Delta) = c\Delta$, where $c$ is the exchange rate. The forward exchange rate is simply $g_2(\Delta) = c$.
The two \emph{paths} through the network, over which we can exchange $A$ for $B$, are given by the two CFMMs.

\paragraph{Optimal routing.} 
The optimal routing problem can be written directly as maximizing the amount of token $B$ received, subject to splitting an input amount $\Delta$ of token $A$ between the two CFMMs:
\[
\begin{aligned}
&\mathrm{maximize} && G_1(\Delta_1) + G_2(\Delta_2)\\
&\mathrm{subject\ to} && \Delta = \Delta_1 + \Delta_2\\
&&& \Delta_1, \Delta_2 \ge 0. 
\end{aligned}
\]
Forward exchange functions are concave, so this problem is a convex optimization problem.
Furthermore, after the optimal trade is made, the marginal forward exchange rates across the two routes will be equal, \ie, $g_1(\Delta_1^*) = g_2(\Delta_2^*)$ where $(\Delta_1^*, \Delta_2^*)$ is a solution to the optimization problem. 
In other words, the optimizer cannot redirect any small amount of flow to another path with a better marginal price.
This fact follows directly from the optimality conditions (see Appendix \ref{app: pigou}).

\paragraph{Selfish routing.} 
In selfish routing, we view the net trade of size $\Delta$ as composed of infinitely many infinitesimal users that act independently.
(We note that atomic routing~\cite{roughgarden2004selfish} may be a more appropriate model but leave exploration of this model in the CFMM context to future work.)
We assume that each path's output is distributed pro-rata to the users trading over that path, which motivates our equilibrium condition: the average price on each route should be equal. On any path in the Pigou network, the average price is given by $\frac{1}{\Delta_i} \int_{0}^{\Delta_i} g_i(t) dt = \frac{1}{\Delta_i} G_i(\Delta_i)$. This leads to the equilibrium equation
\[
\frac{1}{\Delta_1}G_1(\Delta_1) = \frac{1}{\Delta_2}G_2(\Delta_2),
\]
when $\Delta_1$, $\Delta_2 > 0$, subject to the feasibility condition $\Delta_1 + \Delta_2= \Delta$. If one of these paths clearly dominates the other in terms of average price for all flow allocation up to $\Delta$, then all users will choose that path.

\begin{figure}[t]
\captionsetup[sub]{font=scriptsize}
    \centering
    \begin{subfigure}[t]{0.46\textwidth}
        \centering
        \includegraphics[width=\columnwidth]{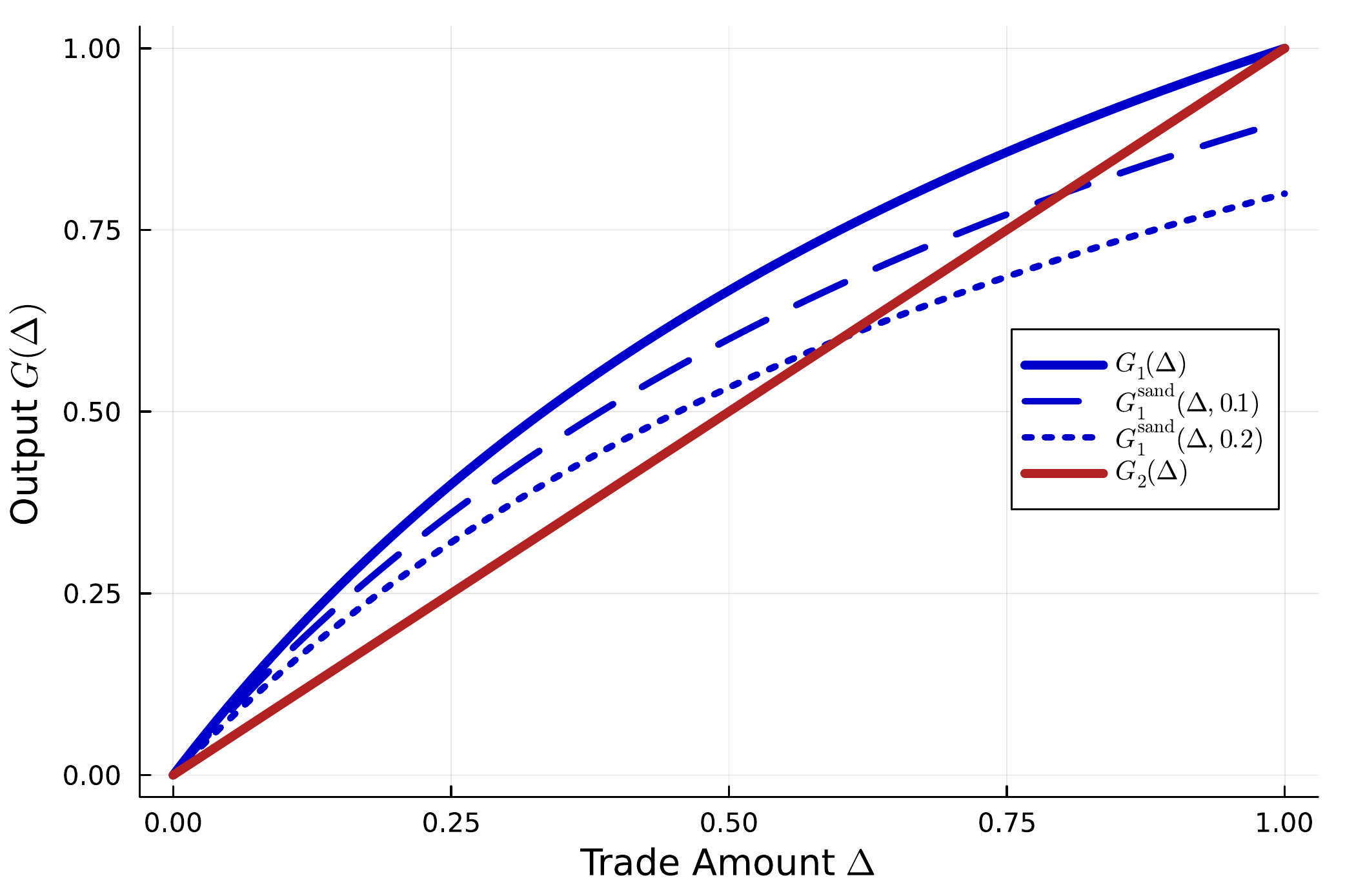}
        \caption{Forward exchange function for each route.}
        \label{fig:pigou-forward-exchange}
    \end{subfigure}
    \hfill
    \begin{subfigure}[t]{0.46\textwidth}
        \centering
        \includegraphics[width=\columnwidth]{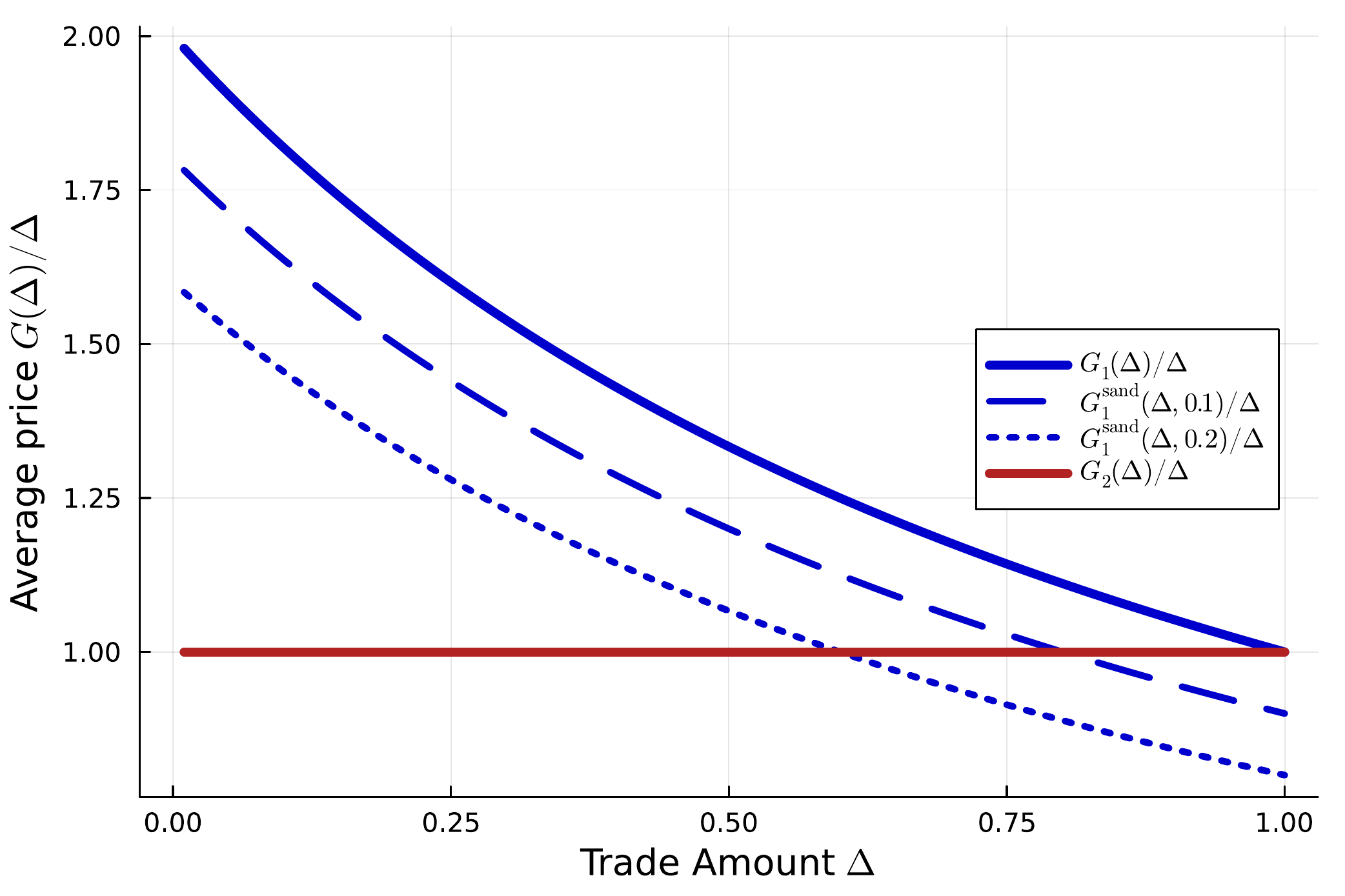}
        \caption{Average price for each route.}
        \label{fig:pigou-price}
    \end{subfigure}
    \caption{Forward exchange function and average price for the Pigou network example.}
    \label{fig:pigou-forward-exchange-and-price}
\end{figure}

\paragraph{Sandwiching.}
Now, we introduce a sandwich attacker on path $1$, the constant product CFMM. 
We denote the forward exchange function with sandwiching by $G_1^\mathrm{sand}(\Delta, \eta)$, which is equal to 
\[
G_1^\mathrm{sand}(\Delta, \eta) = -\frac{(R + \Delta^\mathrm{sand})(R' - G_1(\Delta^\mathrm{sand})}{R + \Delta^\mathrm{sand} + \Delta} +  R' - G_1(\Delta^\mathrm{sand}),
\]
where the optimal sandwich trade is
\[
\Delta^\mathrm{sand} = (1/2)\left( -(\Delta + 2R) + \sqrt{(\Delta + 2R)^2 + 4(R^2 + R\Delta)\left(\tfrac{\eta}{1-\eta}\right)} \right).
\]
We can compute the optimal and selfish routes in the presence of sandwiching by simply replacing $G_1$ with $G_1^\mathrm{sand}$ in the optimal routing problem and in the equilibrium conditions respectively.

\paragraph{Numerical example.}
We consider an instance of this Pigou network with reserves $(R, R') = (1, 2)$ in the constant product CFMM and $c = 1$ for the constant sum CFMM. Consider one unit of token $A$ traded to token $B$, with 
$\Delta_1$ traded through CFMM $1$ and $1-\Delta_1$ through CFMM $2$. 
Without the presence of the sandwich attacker, the optimal route is the $x$ such that
\[
g_1'(\Delta_1^\star) = g_2'(1-\Delta_1^\star) \implies \frac{2}{(1 + \Delta_1^\star)^2} = 1 \implies \Delta_1^\star = \sqrt{2} - 1,
\]
which gives a total output of $G_1(\Delta_1^\star) + G_2(1 - \Delta_1^\star) = 4 - 2\sqrt{2} \approx 1.17$.
The equilibrium, on the other hand, is for all users to use CFMM $1$ (see Figure~\ref{fig:pigou-price}), which has a total output of $1$.
With sandwiching, the top path becomes less desirable. We plot the forward exchange functions and the average price for a number of $\eta$'s in Figure~\ref{fig:pigou-forward-exchange-and-price} as a function of the trade size. (The average price here also a forward exchange rate, where higher is better.) It is clear that the equilibrium will move towards a more balanced split as the slippage tolerance increases, since order flow will move away from CFMM $1$ (see Figure~\ref{fig:pigou-g1-fraction}).
In Figure~\ref{fig:pigou-output} we show that increasing the slippage tolerance hurts optimal routing until it has the same output as selfish routing, at which point all of the order flow goes over CFMM $2$.
It follows that the price of anarchy for this network is highest without the sandwich attacker and decreases to $1$ once neither the optimal nor the selfish routes use CFMM $1$ (see Figure~\ref{fig:pigou-total-output}).
We also note that the profitability of the sandwich attacker increases and then decreases due to the competing effects of increasing price slippage but decreasing order flow over CFMM $1$ (see Figure~\ref{fig:pigou-pnl}).

\begin{figure}[t]
\captionsetup[sub]{font=scriptsize}
    \centering
    \begin{subfigure}[t]{0.32\textwidth}
        \centering
        \includegraphics[width=\columnwidth]{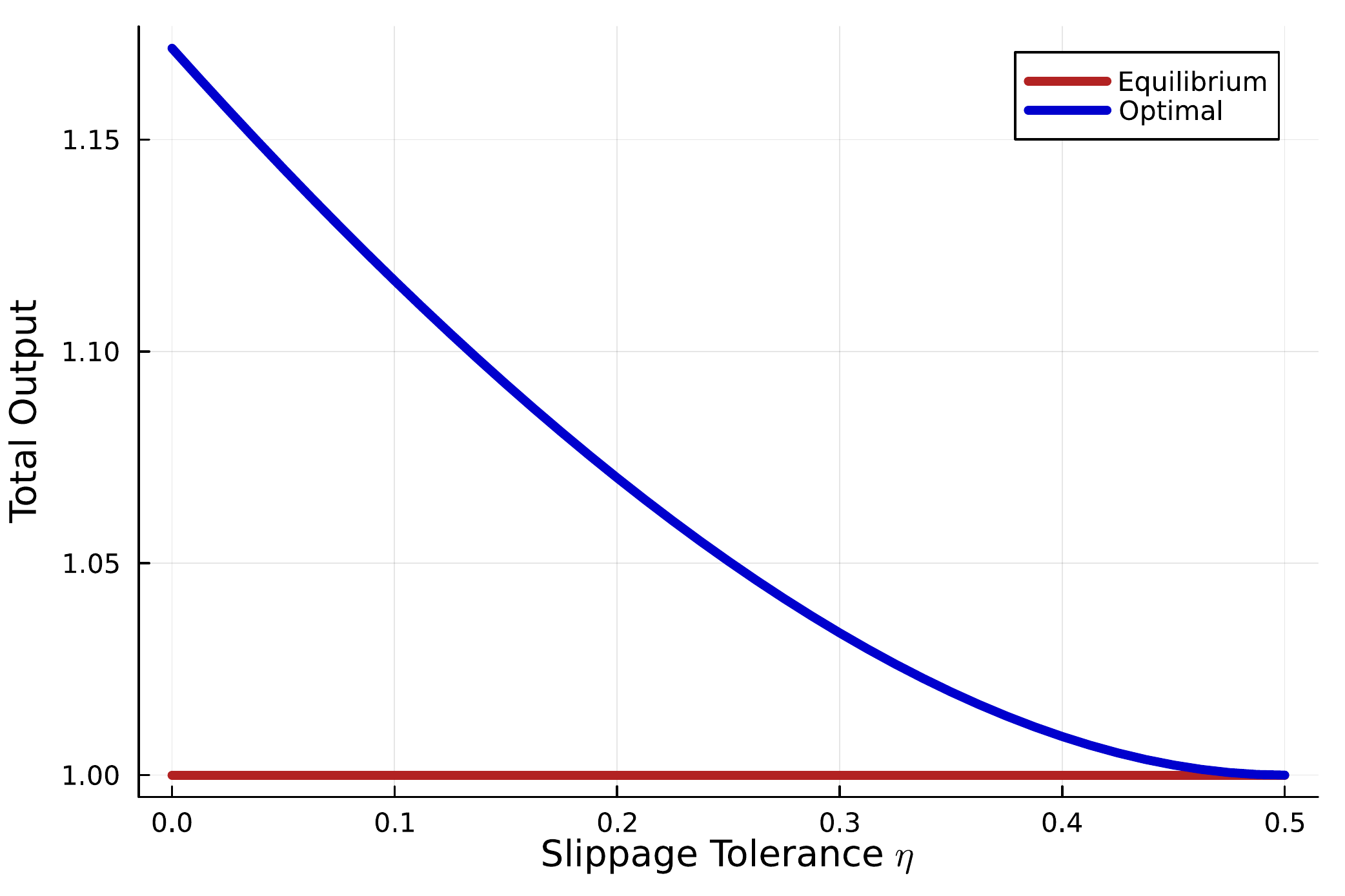}
        \caption{Total output.}
        \label{fig:pigou-total-output}
    \end{subfigure}
    \hfill
    \begin{subfigure}[t]{0.32\textwidth}
        \centering
        \includegraphics[width=\columnwidth]{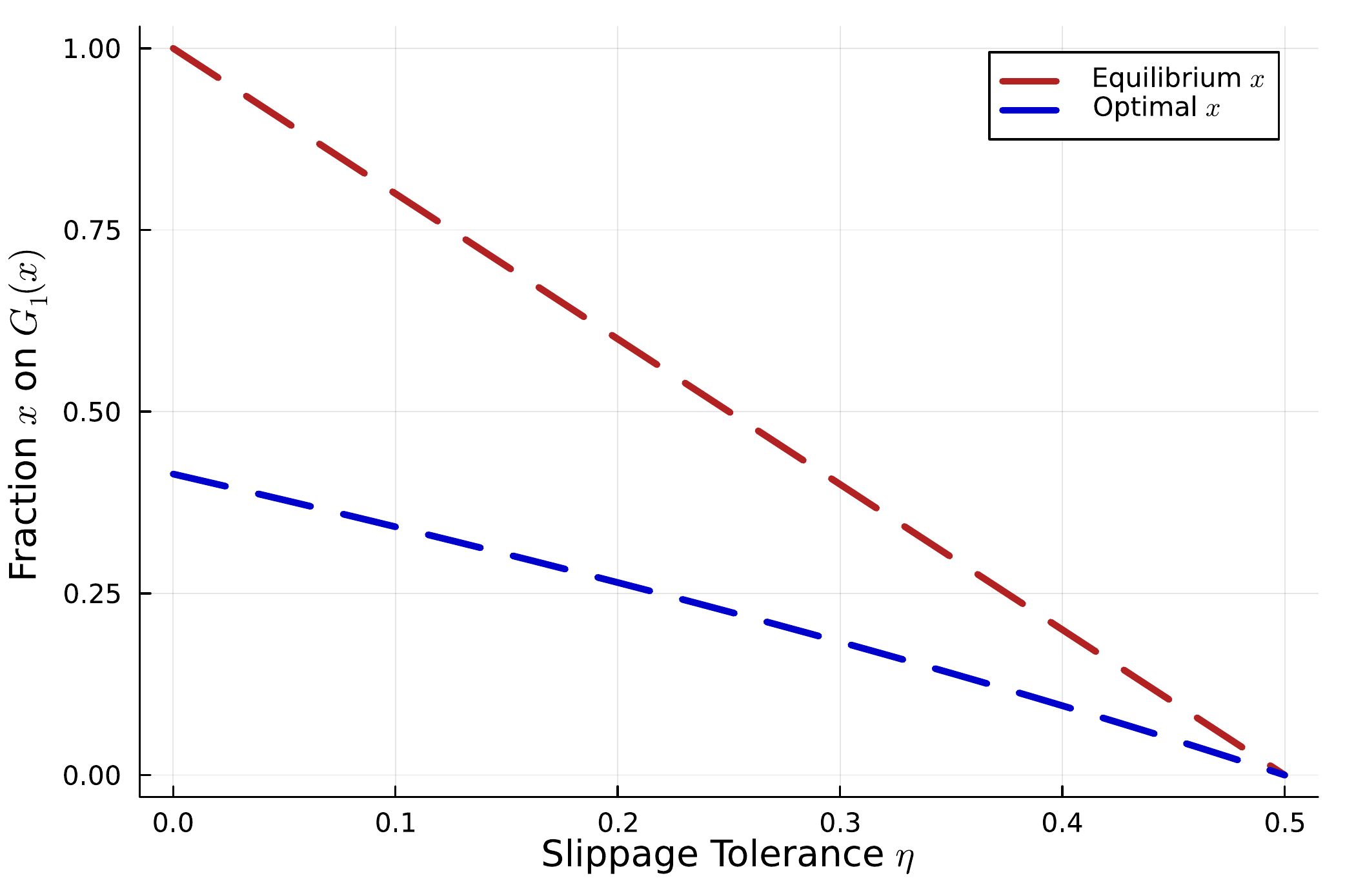}
        \caption{Proportion put on CFMM $1$.}
        \label{fig:pigou-g1-fraction}
    \end{subfigure}
    \hfill
    \begin{subfigure}[t]{0.32\textwidth}
        \centering
        \includegraphics[width=\columnwidth]{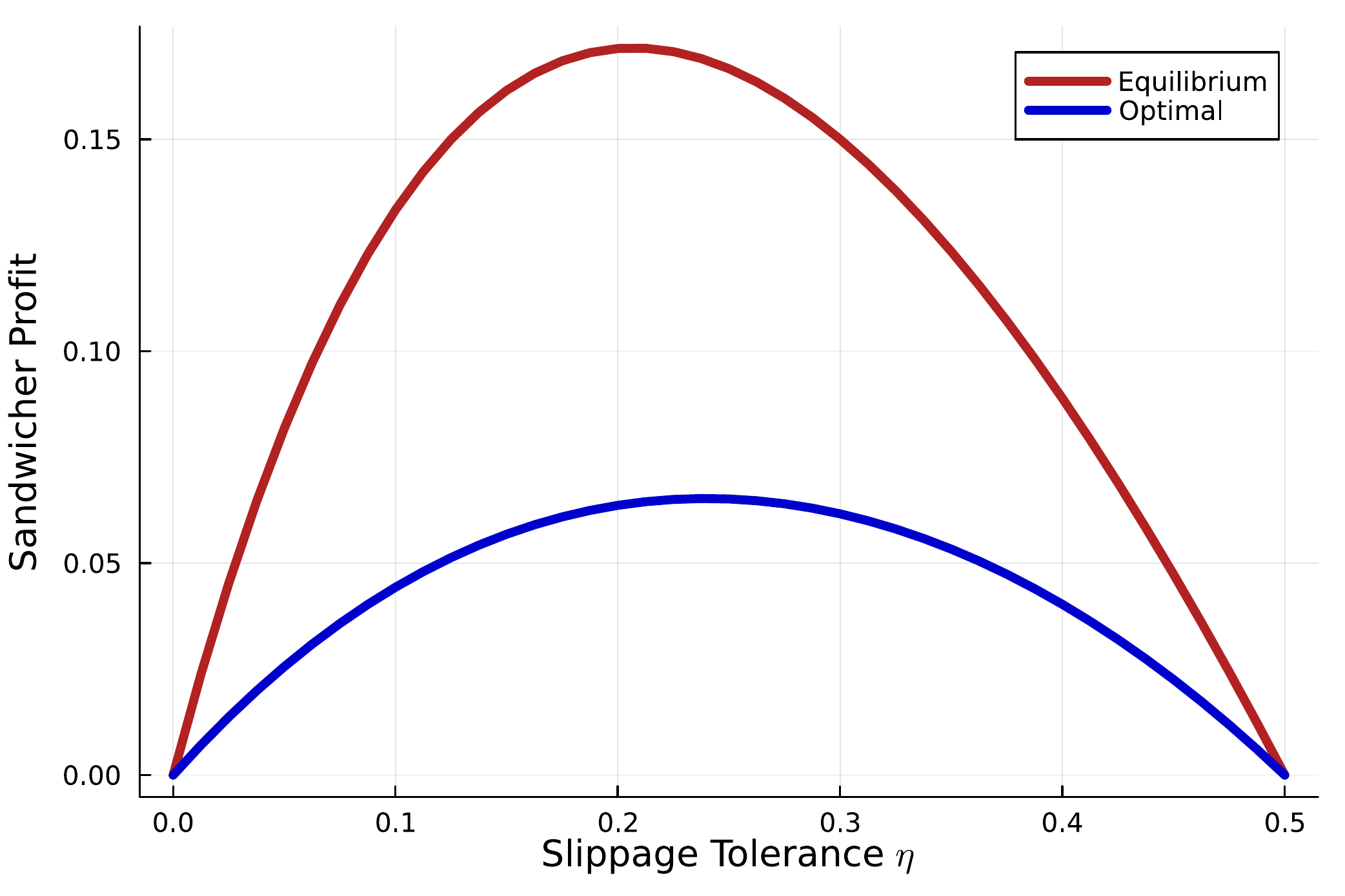}
        \caption{Sandwich attacker profit.}
        \label{fig:pigou-pnl}
    \end{subfigure}
    \caption{Optimal and selfish routing in terms of total output, route taken, and sandwich attacker profit in the Pigou network as the slippage tolerance $\eta$ varies.}
    \label{fig:pigou-output}
\end{figure}

\subsection{The CFMM Braess Example}\label{sec:braess}

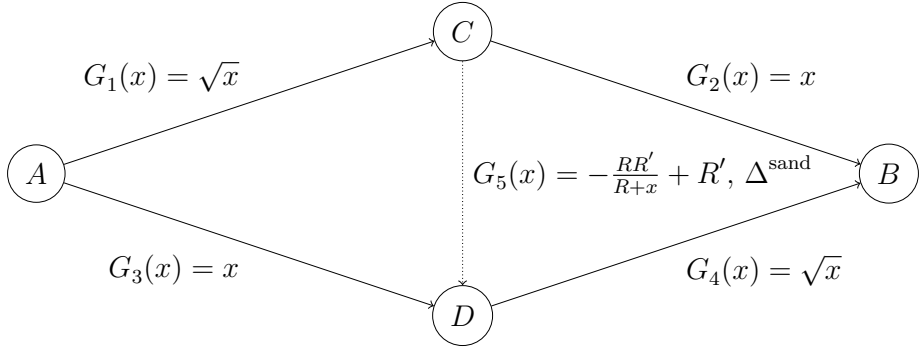
\begin{figure}[h!]
\centering
\resizebox{0.75\linewidth}{!}{
\begin{tikzpicture}
\node[draw, circle] (A) at (0,0) {$A$};
\node[draw, circle] (B) at (12,0) {$B$};
\node[draw, circle] (C) at (6,2) {$C$};
\node[draw, circle] (D) at (6,-2) {$D$};

\draw[->] (A) -- (C);
\draw[->] (A) -- (D);
\draw[->] (C) -- (B);
\draw[->] (D) -- (B);
\draw[->, densely dotted] (C) -- (D);

\node[above left] at ($.5*(A) + .5*(C)$) {$G_1(x) = \sqrt{x}$};
\node[above right] at ($.5*(B) + .5*(C)$) {$G_2(x) = x$};
\node[below left] at ($.5*(A) + .5*(D)$) {$G_3(x) = x$};
\node[below right] at ($.5*(B) + .5*(D)$) {$G_4(x) = \sqrt{x}$};
\node[right] at ($.5*(C) + .5*(D)$) {$G_5(x) =  -\frac{RR'}{R + x} + R'$, $\ds$};

\end{tikzpicture}
}
\caption{The CFMM Braess Network.}
\label{fig:braess}
\end{figure}

\paragraph{Sandwich attackers can improve routing.}
We now construct an example representing an `inverse Braess paradox' in CFMMs in which the presence of sandwich attackers improves the quality of selfish routing by reducing the price of anarchy.
This example demonstrates that while sandwich attacks necessarily worsen the execution of individual trades on a given CFMM, they can improve the \emph{social welfare}, measured as the combined output that users receive over all paths under selfish routing.
For the CFMM network in Figure~\ref{fig:braess}, we introduce the variables $\Delta_i \in \reals$, $i= 1, \dots, 5$ to denote the input into CFMM $i$. 
The output of CFMM $i$ is then $G_i(\Delta_i)$.
Analogously to the CFMM Pigou example, we will consider both optimal and selfish routing.

\paragraph{No middle CFMM.}
First we consider the simple case where $G_5(\Delta)$ does not exist; all users must use either the top or the bottom route.
In this case, the top and bottom routes are equivalent, and the output of both is strictly concave.
As a result, users routing selfishly will split their order flow equally between the
two routes. It is easy to see that this equilibrium is, in fact, optimal as well.

\paragraph{Optimal routing.} 
Next we consider optimal routing in the case with the middle CFMM.
We wish to maximize the output at node $B$, \ie,  $G_4(\Delta_4) + G_2(\Delta_2)$, subject to the flow conservation constraints implied by this network.
We can then find an optimal trade split by solving the following optimization problem:
\[
\begin{aligned}
&\mathrm{maximize} && G_4(\Delta_4) + G_2(\Delta_2)\\
&\mathrm{subject\ to} && 1 = \Delta_1 + \Delta_3\\
&&& G_1(\Delta_1) = \Delta_2 + \Delta_5 \\
&&& G_3(\Delta_3) + G_5(\Delta_5) = \Delta_4 \\
&&& \Delta_i \ge 0 \text{ for } i = 1,\dots,5.
\end{aligned}
\]
In fact, since all the CFMM functions $G_i$ are concave and increasing, the convex relaxation of this problem, found by relaxing the equality constraints into inequalities, is tight. 
Furthermore, at optimality, the marginal forward exchange rate across each route will be equal.
See~\cite{angeris2022optimal, diamandis2023efficient} for details and further discussion of optimal routing.

\paragraph{Selfish routing.}
Recall that the equilibrium condition is that the average price across all used paths is equal, and that the trades are feasible.
Let $\alpha_1$, $\alpha_2$, and $\alpha_3$ be the order flow on the top ($G_1 \to G_2$), bottom ($G_3 \to G_4$), and middle ($G_1 \to G_5 \to G_4$) paths respectively.
Again, we assume that, after a trade has been made, the output is distributed pro-rata to all users of a particular path.
The average price condition can be written as
\[
\begin{aligned}
    \frac{1}{\alpha_1}\cdot G_2\left(\frac{\alpha_1}{\alpha_1 + \alpha_3} \Lambda_1 \right)
    &= \frac{1}{\alpha_2} \cdot \frac{\Lambda_3}{\Lambda_3 + \Lambda_5} \cdot G_4(\Lambda_5 + \Lambda_3)
    &= \frac{1}{\alpha_3} \cdot \frac{\Lambda_5}{\Lambda_3 + \Lambda_5} \cdot G_4(\Lambda_5 + \Lambda_3),
\end{aligned}
\]
where $\Lambda_i$ is the output of CFMM $i$, for example $\Lambda_1 = G_1(\alpha_1 + \alpha_3).$
(Here we assume all paths are used. A more complete definition is given in the sequel.)
The feasibility conditions are the same as those for optimal routing:
\[
\begin{aligned}
1 &= \alpha_1 + \alpha_2 + \alpha_3\\
\Delta_1 &= \alpha_1 + \alpha_2 \\
\Delta_3 &= \alpha_3\\
\Lambda_1 &= \Delta_2 + \Delta_5 \\
\Lambda_3 + \Lambda_5 &= \Delta_4 \\
\Lambda_i &= G_i(\Delta_i), \quad \text{ for } i = 1,\dots,5 \\
\Delta_i &\ge 0, \quad \text{ for } i = 1,\dots,5. 
\end{aligned}
\]
The solution to this system is an equilibrium; no infinitesimal flow can get a better output for its share of the flow (\ie, a better average price) by deviating. 

\begin{figure}[t]
\captionsetup[sub]{font=scriptsize}
    \centering
    \begin{subfigure}[t]{0.46\textwidth}
        \centering
        \includegraphics[width=\columnwidth]{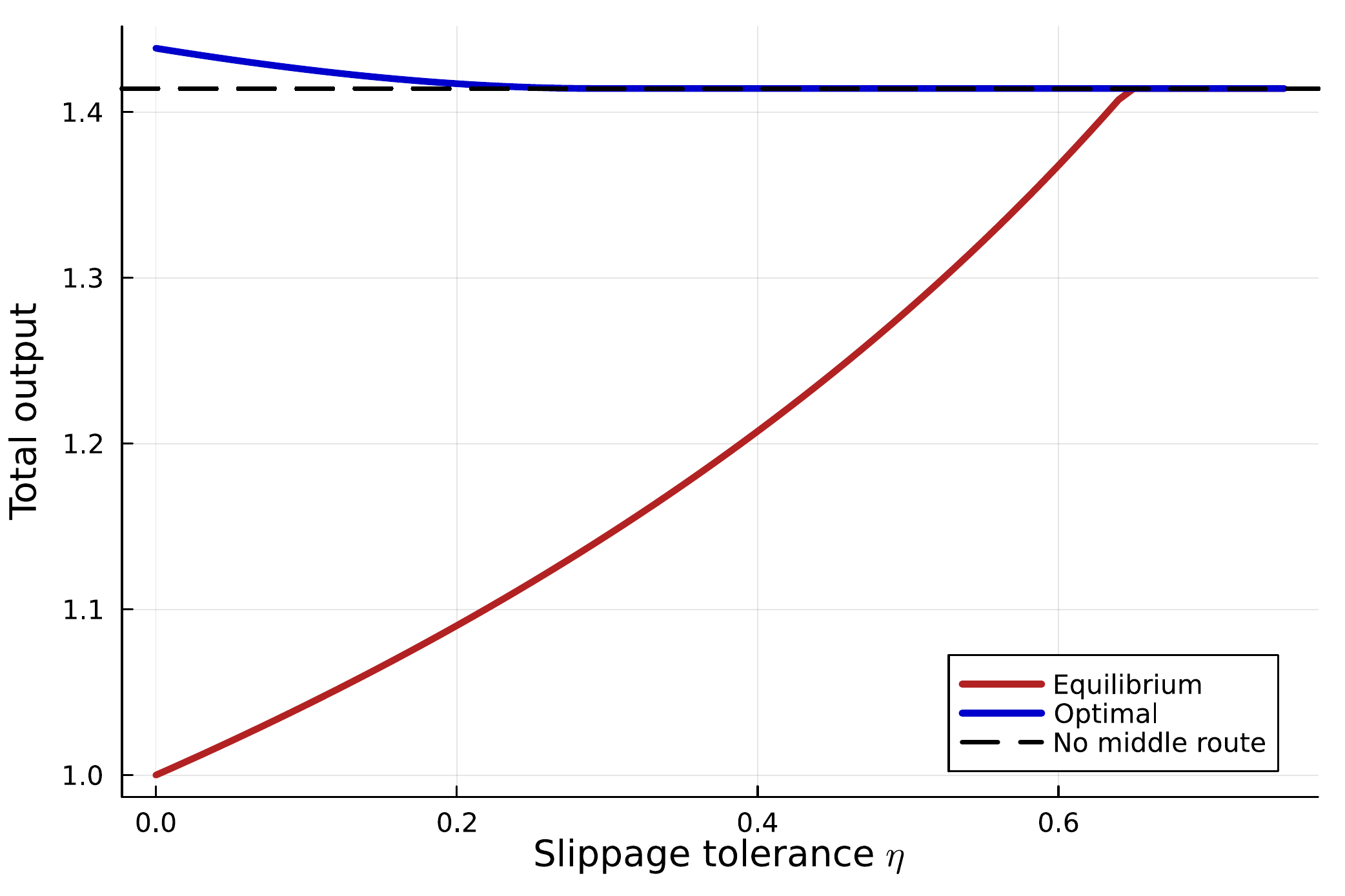}
        \caption{Output for optimal and equilibrium routes.}
        \label{fig:braess-output}
    \end{subfigure}
    \hfill
    \begin{subfigure}[t]{0.46\textwidth}
        \centering
        \includegraphics[width=\columnwidth]{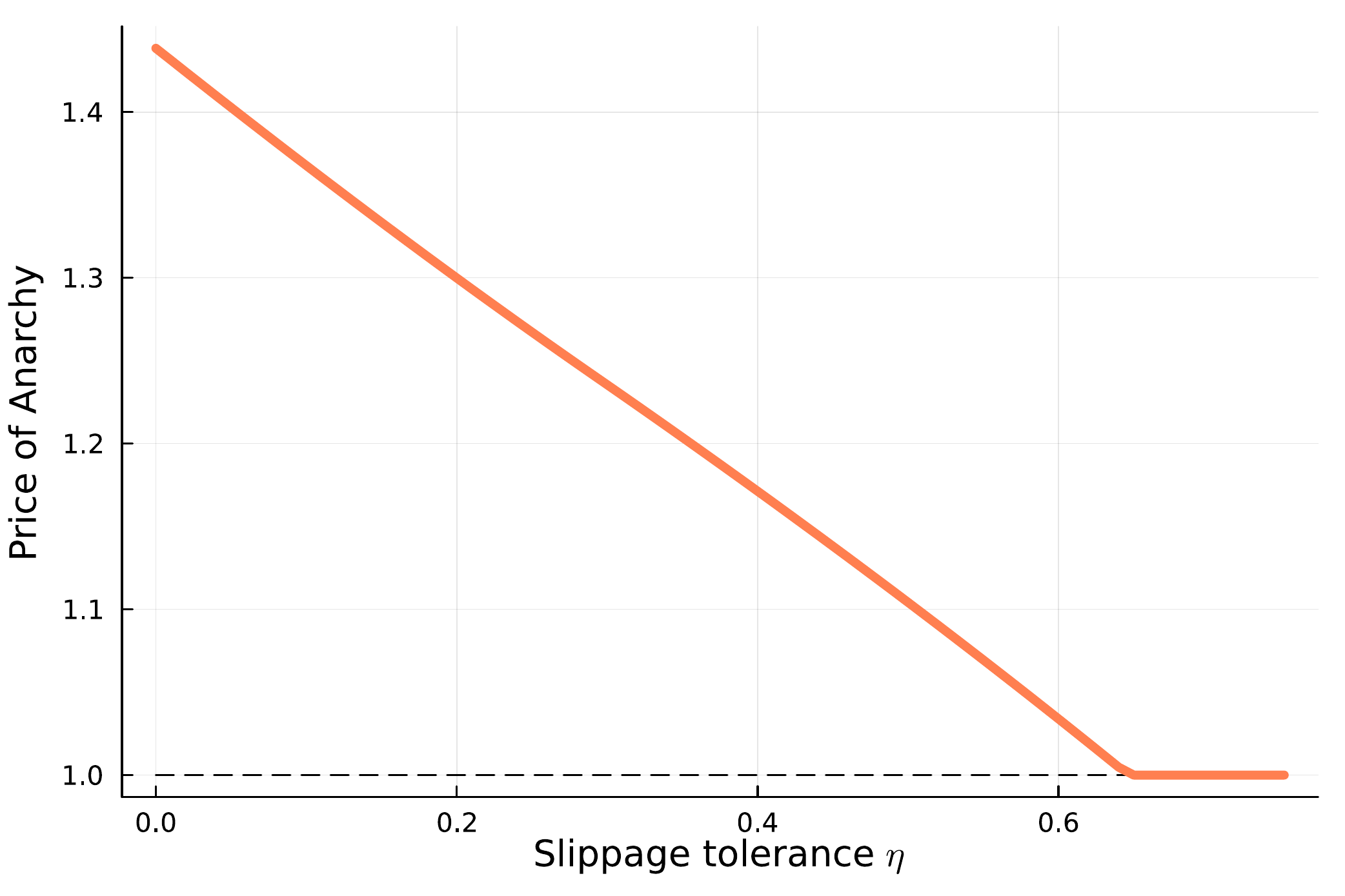}
        \caption{Price of anarchy for the Braess network.}
        \label{fig:braess-poa}
    \end{subfigure}
    \caption{Output and price of anarchy in the Braess network example.}
\end{figure}

\paragraph{Numerical example.}
We consider an instance of this Braess network with middle CFMM reserves $(R, R') = (1, 2)$ and one unit of token $A$ traded to token $B$.
Without the middle CFMM, the optimal and equilibrium routes both split the order flow equally between the top and bottom routes, resulting in an output of $2 \cdot \sqrt{0.5} = 1.414$. 
When we add the middle CFMM, the path through this CFMM has the best average price, even when all the flow is allocated to it.
As a result, we observe congestion in selfish routing. 
This effect is clearly illustrated in Figure~\ref{fig:braess-output}.
We see that, despite the addition of the new market, the equilibrium output decreases without sandwiching ($\eta = 0$).
As the slippage tolerance increases and users get a worse execution price, they move away from using the middle CFMM. 
As a result, the optimal routing and selfish routing output converge to the optimal output without the middle CFMM, and the price of anarchy decreases to $1.0$, shown in Figure~\ref{fig:braess-poa}.
Figure~\ref{fig:braess-middle} illustrates the decrease of the fraction of order flow allocated to this path.
In this way, the presence of the sandwich attacker relieves congestion on the network.
The opposing effects of increased slippage but less order flow on CFMM $5$ again cause the sandwich attacker profit to increase and then decrease with slippage tolerance (see Figure~\ref{fig:braess-pnl}).

\begin{figure}[t]
\captionsetup[sub]{font=scriptsize}
    \centering
    \begin{subfigure}[t]{0.46\textwidth}
        \centering
        \includegraphics[width=\columnwidth]{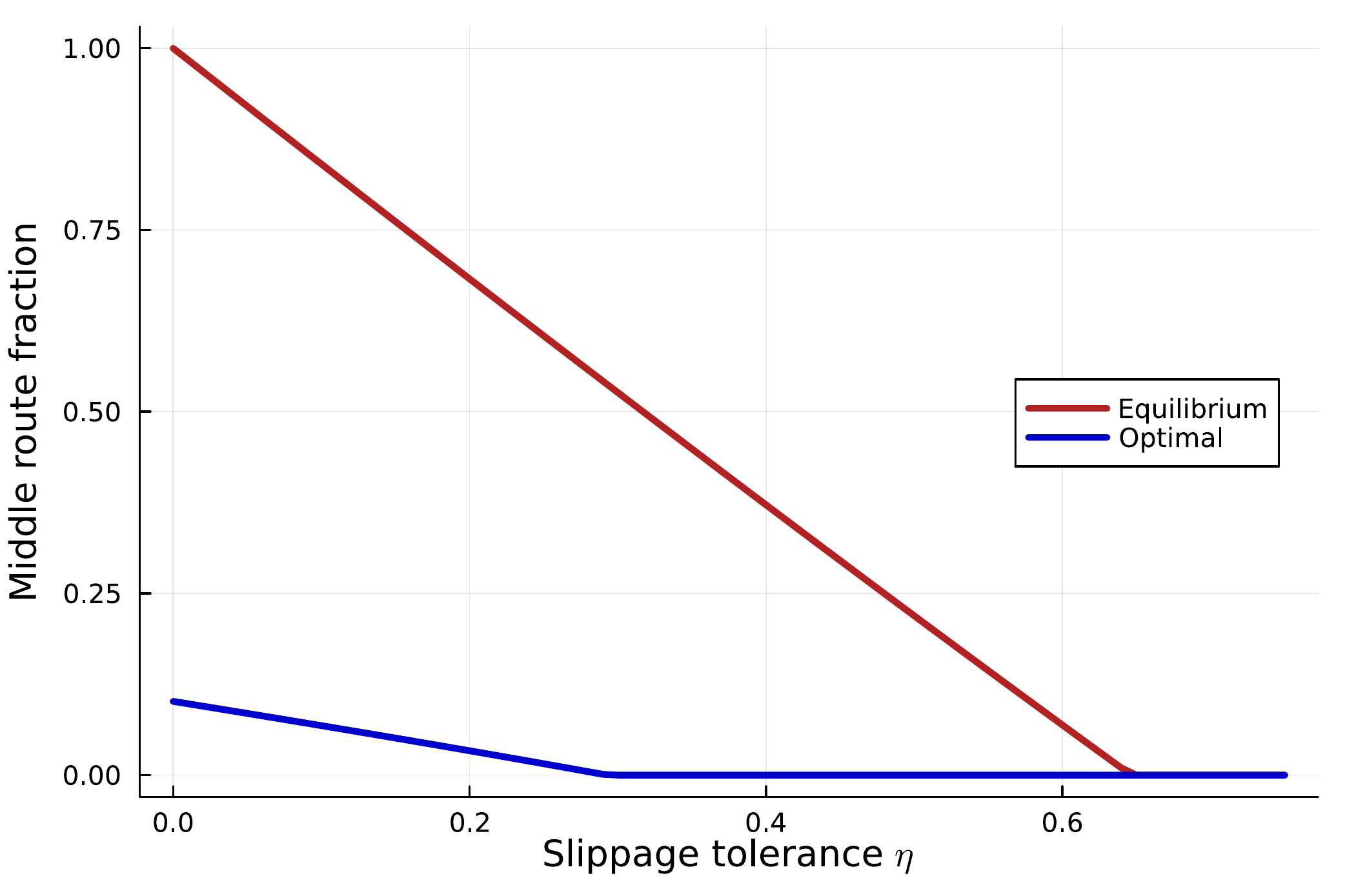}
        \caption{Fraction on the middle link in the Braess network.}
        \label{fig:braess-middle}
    \end{subfigure}
    \hfill
    \begin{subfigure}[t]{0.46\textwidth}
        \centering
        \includegraphics[width=\columnwidth]{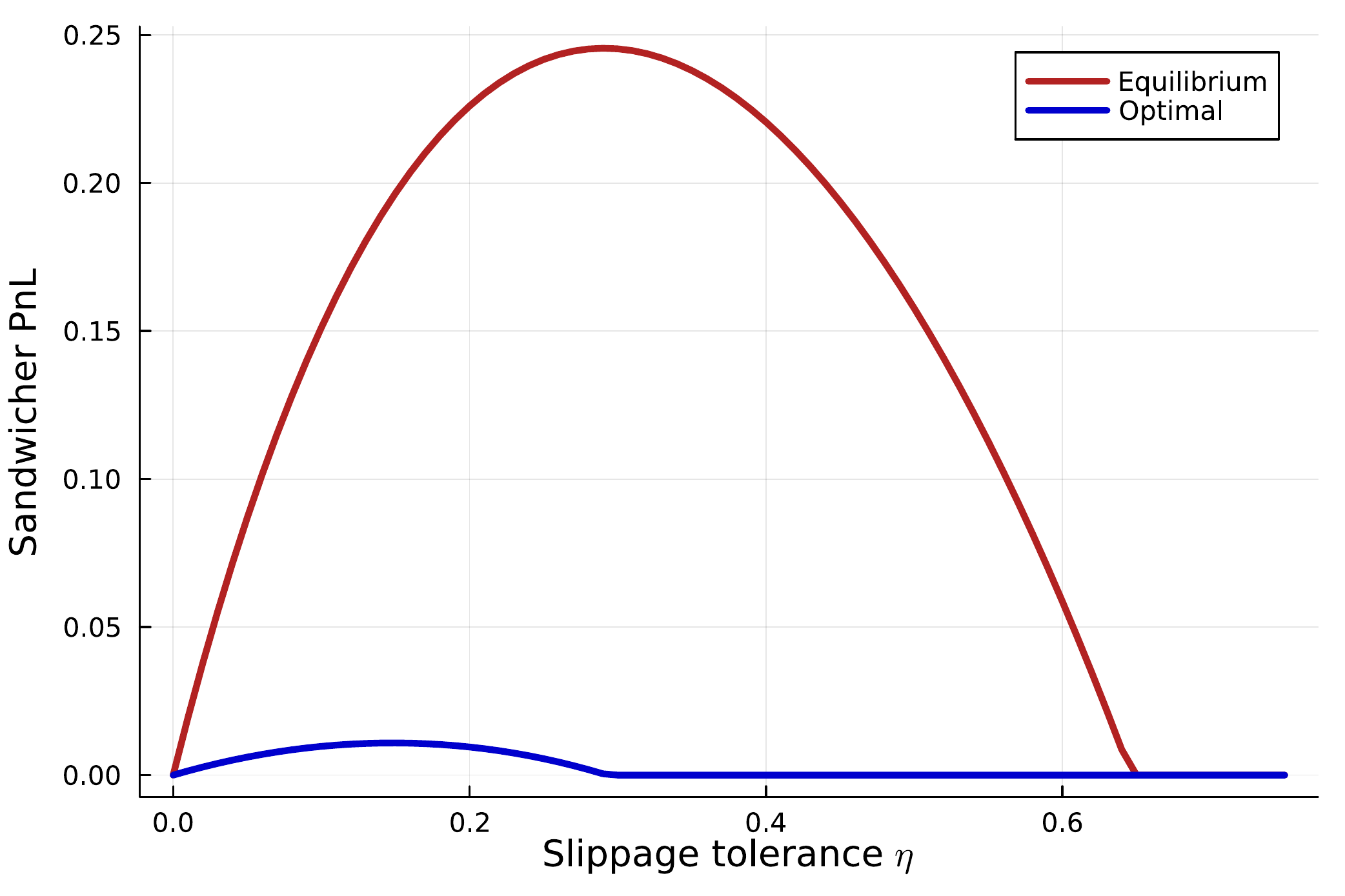}
        \caption{Sandwich attacker profit for the Braess network.}
        \label{fig:braess-pnl}
    \end{subfigure}
    \caption{Characteristics of the middle link route in the Braess network example.}
\end{figure}

\paragraph{Discussion.}
This toy example illustrates an `inverse Braess's paradox'.
Like in Braess's paradox, the addition of the additional CFMM decreases the total output at equilibrium of selfish routing (of course, it increases the optimal output).
All of the selfish order flow wants to use the middle link, which causes congestion.
However, the addition of the sandwich attacker causes this order flow to begin to migrate away from the middle CFMM as users' slippage tolerance increases.
The sandwich attacker can be thought of as a `decentralized traffic controller', where users' best responses to the existence of a sandwich attacker lead to better output across the network. 
We refer to this effect as an `inverse Braess's paradox'.

\subsection{Price of Anarchy}
A general formulation for impact of sandwich attacks on networks of CFMMs follows from the intuition built by the previous two examples. Suppose that we have a graph $G = (V, E)$ where each vertex $A \in V$ denotes a token and each edge $e = (A, B) \in E$ represents a CFMM for trading between tokens $A$ and $B$. Denote the set of paths from $A$ to $B$ as $\mathcal{P}$.
Associated to each $e \in E$ is a price function $g_e(\cdot)$ and a corresponding forward exchange function $G_e(\cdot)$. The output of the users' trade into edge $e$ is defined by the function $G_e(\cdot)$. 
We first define the general setup of \emph{optimal routing}, in which a central planner is able to maximize the net output from a network of CFMMs for a user trading between a pair of tokens, and \emph{selfish routing}, in which infinitesimal users greedily optimize for their pro-rata share of the output along a path. The outcome of this process is an equilibrium.
We then define \emph{routing MEV} as any excess value that can be extracted from adjusting how transactions are executed on this graph (by sandwich attackers). Informally, we prove: 
\begin{center}
    \textit{The price of anarchy of selfish routing through a network of CFMMs is bounded by a constant that depends on the slippage $\eta$ and constants $\kappa, \mu, \alpha, \beta$.}
\end{center}

\paragraph{Trade Splittings and Path Outputs.} To analyze the price of anarchy, we define the space of trades on a network as $\mathcal{T} = \reals_+ \times [0,1]^{|\mathcal{P}|} \times V \times V$. Each trade $T \in \mathcal{T}$ where $T = (\Delta , \eta, A, B)$ specifies an amount, $\Delta$ to be traded from vertex $A$ to vertex $B$ along with a slippage limit $\eta_p$ on every path $p \in \mathcal{P}$. We will abuse notation and utilize $(\Delta_{AB}, \eta_{AB}) \in T$ to refer to $(\Delta, \eta, A,B) \in T$.
For a trade $(\Delta_{AB}, \eta_{AB})$, we define any $\alpha \in \{x \in \reals^{|\mathcal{P}|} : \sum_i x_i = \Delta_{AB} \} = 
\hat{S}_{|\mathcal{P}|}$ to be a \emph{splitting vector} that indicates what fraction of the $\Delta_{AB}$ units of token $A$ are routed onto each path. 
For every path $p$, we denote the \emph{path forward exchange} function $G_p : \hat{S}_{|\mathcal{P}|} \rightarrow \reals$. 
This function gives the output of an amount of the trade $\alpha_p \Delta_{AB}$ placed on path $p \in \mathcal{P}$. 
We use this function as opposed to the edge forward exchange functions $G_e(\cdot)$ because the conditions for optimal and selfish routing are tractably written in this notation. 
We note that for any trade splitting $\alpha$, the path output function $G_p(\cdot)$ may depend on components of $\alpha$ belonging to paths other than path $p$, as there may be edges that intersect on multiple paths.  

\paragraph{Optimal Routing.} We define optimal routing over the trade split $\alpha$. (Note that this formulation differs from the original formulation of~\cite{angeris2022optimal}). In terms of this trade split $\alpha$ and the path output function $G_p$, the optimal routing problem is
\[
\begin{aligned}
&\mathrm{maximize} && \sum_{p \in \mathcal{P}} G_p(\alpha) \\
&\mathrm{subject\ to} && \sum_{p \in \mathcal{P}} \alpha_p = \Delta_{AB} \\
&&& \alpha_p \geq 0. 
\end{aligned}
\]
We denote any trade splitting that is a solution to the optimal routing problem by $\alpha^{\star}$. A modified version of the objective $\sum_{p \in \mathcal{P}} G_p(\alpha)$ (after accounting for the presence of sandwich attacks on the network) will be used when defining the price of anarchy. 

\paragraph{Selfish Routing.} 
In order to generalize the notion of selfish routing in Section \ref{sec:pigou}, we define an equilibrium in terms of the average price on each path. Once again, the interpretation of our equilibrium notion is that a small unit of flow should not want to deviate from the path it has chosen because it has greedily optimized for the share of the output it will receive.
A splitting vector $\bar{\alpha} \in \hat{S}_{|\mathcal{P}|}$ is said to be an \textit{equilibrium splitting}\footnote{
    The existence of such an equilibrium can be established using standard results in nonatomic routing games~\cite{roughgarden2005selfish}.
} 
if for all $p,p' \in \mathcal{P}$ with $\bar{\alpha}_p > 0$ we have 
\begin{align*}
    \frac{G_p(\bar{\alpha})}{\bar{\alpha}_p}  \geq \frac{G_{p'}(\bar{\alpha})}{\bar{\alpha}_{p'}} 
\end{align*}
\noindent This definition says that over \emph{all} paths on which there is a nonzero trade in equilibrium, the average price must be equal. This equilibrium condition says that an infinitesimal trade (or unit of `flow') on a path from $A$ to $B$ greedily optimizes for its pro-rata share of the output coming from that path. 
This situation stands in stark contrast to optimal routing, which tries to maximize the net output from the network and requires the existence of a central planner who can route trades accordingly. 

\paragraph{Resolving ambiguities in splitting.} We note that the selfish routing equilibrium over paths described above leaves trades splitting at intermediate nodes in the network ambiguous, so we provide a procedure for resolving this ambiguity. 
Consider a trade $\Lambda$, made up of two previous CFMM outputs, that then must be split among two CFMMs, as in Figure \ref{fig:ambiguity}.
Suppose that trades $\Delta_1$ and $\Delta_2$, corresponding to each path, are incident on $G_1$, and so this CFMM outputs $\Lambda = G_1(\Delta_1 + \Delta_2)$. 
We prescribe that this trade is inputted into CFMMs $G_2$ and $G_3$ by splitting it according to the pro-rata percentage of the input into $G_1$ it represents. 
That is, the input into the top path is $\frac{\Lambda \Delta_1}{\Delta_1 + \Delta_2}$ and the input into the bottom path is $\frac{\Lambda \Delta_2}{\Delta_1 + \Delta_2}$.

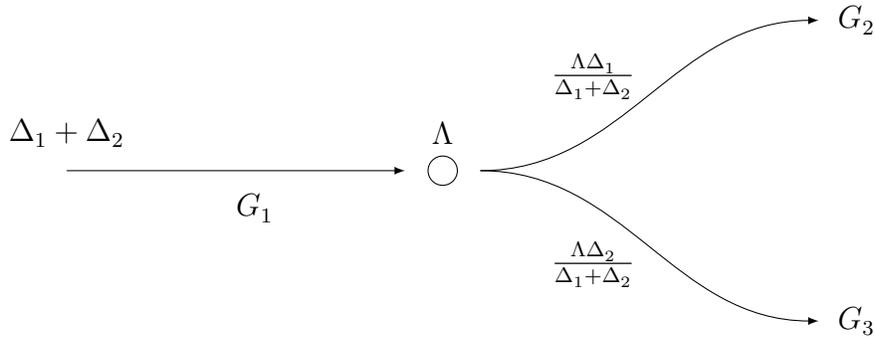
\begin{figure}[h]
\centering
\begin{tikzpicture}
\node[draw, circle] at (0,0) {};
\draw[->, >=latex] (-5, 0) to[out=0, in=180] ( -0.5,0);
\draw[->, >=latex] (0.5, 0) to[out=0, in=180] (5,2);
\draw[->, >=latex] (0.5, 0) to[out=0, in=180] (5,-2);
\node at (-2.5, -0.5) {$G_1$};
\node at (5.5, 2) {$G_2$};
\node at (5.5, -2) {$G_3$};
\node at (-5, 0.5) {$\Delta_1 + \Delta_2$};
\node at (0, 0.5) {$\Lambda$};
\node at (2, 1.25) {$\frac{\Lambda \Delta_1}{\Delta_1 + \Delta_2}$};
\node at (2, -1.25) {$\frac{\Lambda \Delta_2}{\Delta_1 + \Delta_2}$};
\end{tikzpicture}
\caption{Resolving ambiguous splits in CFMM networks.}
\label{fig:ambiguity}
\end{figure}

\paragraph{Sandwich Attacks on Graphs.} We now move to defining optimal and selfish routing in the presence of sandwich attackers by modifying the path output functions $G_p^{\text{sand}}(\cdot)$ to account for the value captured by sandwich attackers on the CFMM networks. Recall that users provide a slippage limit $\eta_{AB} \in [0,1]^{|\mathcal{P}|}$ over all the paths. However, unlike the single-edge case of~\S\ref{sec:sandwich}, we need to construct slippage limits for each edge in a path. We solve for the implied slippage limits on each edge in a path to bound sandwich attack size and price impact. 
Specifically, for a path $p = (e_1, \ldots, e_{|p|}) \in \mathcal{P}$ along with a slippage $\eta_{AB,p}$, we know that for any trade splitting $\alpha \in \hat{S}_{|\mathcal{P}|}$, the minimum acceptable amount from path $p$ is $(1-\eta_{AB,p}) G_p(\alpha)$. First, we denote for a path $p$, $\Delta_{e_{|p| - k}}$ to be the net trade that enters edge $|p|-k$ on that path, given a trade $\Delta_{AB}$ and splitting  $\alpha$. Then, we can define the edge slippage limits for path $p$, given by the vector $(\eta_{p,e_1}, \ldots, \eta_{p,|p|})$, where $\eta_{p,|p|}  = \eta_{AB,p}$ recursively as:
\begin{equation}
\label{eq: recursion}
\begin{aligned}
    G_{e_{|p|-1}}
    (
        \de_{e_{|p|-1}} 
        + \ds_{e_{|p|-1}}
    ) 
    - G_{e_{|p|-1}}
    (\de_{e_{|p|-1}}) 
    &= (1-\eta_{{AB},p}) 
    G_{e_{|p|}}(\de_{|p|})  \\  
    G_{e_{|p| - k -1}}(\Delta_{e_{|p| - k - 1}}
    + \ds_{e_{|p|-k-1}}) 
    - G_{e_{|p|-k-1}}(\de_{e_{|p|-k-1}}) 
    &= (1-\eta_{p,k}) G_{e_{|p|-k}}(\de_{|p|-k})
\end{aligned}
\end{equation}
The above set of equations for $k = 1,\dots, |p|-1$ provide recursions that can be solved from the terminal edge back to edge $e_1$ to find the corresponding slippage limits and sandwich amounts on every path $p \in \mathcal{P}$. Once these slippage limits have been pinned down on every edge, there is a well-defined notion of a \emph{path} sandwich attack $\ds_p(\alpha, \eta_{AB})$. Therefore, we can define the modified path output functions $G_p^{\text{sand}}(\alpha, \eta_{AB})$ as the amount of output that the user receives over a path after the sandwich attacker's excess output over that path has been removed: 
\begin{align*}
    G_{p}^{\text{sand}}(\alpha, \eta_{AB}) \coloneqq G_p(\alpha_{p} + \ds_p(\alpha, \eta_{AB})) - G_p(\ds_p(\alpha, \eta_{AB})) 
\end{align*}

Next, we note that the recursion for solving for $\eta_{p,e_i}$ given in~\eqref{eq: recursion} is not guaranteed to return unique slippage limits as it is the difference of convex functions, which means there could be multiple sequences $\eta_{p,e_1}, \ldots, \eta_{p,|p|}$that could lead to the same optimal output. Using curvature bounds, however, we can bound this recursion and, therefore, provide bounds on $\ds_p$. 

\paragraph{Bounds on $\ds_p$.} In order to derive bounds on the price of anarchy, we must bound the size of a sandwich attack along a path $\ds_p(\alpha,\eta_{AB})$ by the trade size on that path and by the terminal slippages.  The following result says that the size of the sandwich attack along a path $p$ is bounded on both sides by functions that depend on the path length, and the trade size $\alpha_p$. These functions will be used to derive the upper bound on price of anarchy in Theorem \ref{thm2}: 

\begin{prop}\label{prop6}
There exist functions $f(\kappa, \mu, \eta_{AB})$ and $g(\kappa, \mu,\eta_{AB})$ such that 
\begin{align}\label{eq:ds-edge-bd}
    \ds_p(\alpha, \eta_{AB}) \leq f(\kappa, \mu, \eta_{AB})^{|p|} \alpha_p
\end{align}
and
\begin{align}\label{eq:d-edge-bd}
    \ds_p(\alpha, \eta_{AB}) \geq g(\kappa, \mu , \eta_{AB})^{|p|} \alpha_p
\end{align}
for all paths $p \in \mathcal{P}$. 
\end{prop}
Using the bounds in Proposition \ref{prop6} and the $(\mu,\kappa)$-smoothness of $G_p(\cdot)$, we immediately get the following bounds on $G_{p}^{\text{sand}}(\alpha, \eta_{AB})$: 
\begin{align*}
    G_{p}^{\text{sand}}(\alpha, \eta_{AB}) \leq \left( \mu + \mu f(\kappa, \mu, \eta_{AB})^{|p|} - \kappa g(\kappa, \mu, \eta_{AB})^{|p|} \right) \alpha_p \\ 
    G_{p}^{\text{sand}}(\alpha, \eta_{AB}) \geq \left( \kappa + \kappa g(\kappa, \mu, \eta_{AB})^{|p|} - \mu f(\kappa, \mu, \eta_{AB})^{|p|} \right) \alpha_p
\end{align*}
Further, the function $G_p$ also has local quadratic bounds (that is, there exist $\kappa', \alpha'$ such that $\kappa' \alpha_p^2 \leq G_p(\alpha) \leq \mu' \alpha_p^2$). We can use these bounds to derive analogous quadratic bounds for $G_p^{\text{sand}}(\alpha,\eta_{AB})$: 
\begin{align*}
    G_{p}^{\text{sand}}(\alpha, \eta_{AB}) \leq \left( \mu + \mu f(\kappa, \mu, \eta_{AB})^{|p|} - \kappa g(\kappa, \mu, \eta_{AB})^{|p|} \right) \alpha_p^2 \\ 
    G_{p}^{\text{sand}}(\alpha, \eta_{AB}) \geq \left( \kappa + \kappa g(\kappa, \mu, \eta_{AB})^{|p|} - \mu f(\kappa, \mu, \eta_{AB})^{|p|} \right) \alpha_p^2
\end{align*}
The definition of the modified path output function and the bounds provided on $\ds_p(\alpha,\eta_{AB})$ will be important in defining the excess price impact realized by the user and the sandwich profit realized by attacker.
Proposition \ref{prop6} (and associated desiderata about analyzing the aforementioned recursions) is described in Appendix \ref{app:prop6}.

\paragraph{Social Welfare and the Price of Anarchy.} To define a measure of social welfare, we compare the net output under selfish routing to the net output under optimal routing in the presence of sandwichers. 
This quantity indicates how much sandwich attackers degrade the quality of selfish routing relative to optimal routing. Our welfare function takes the sum of the path output functions incorporating sandwiches, $W^{\text{sand}}(\alpha) = \sum_{p \in \mathcal{P}} G_p^{\text{sand}}(\alpha)$, which allows us to formally define the \emph{price of anarchy} for any optimal splitting $\alpha^*$ and any equilibrium splitting $\bar{\alpha}$: 
\begin{align*}
    \textsf{PoA}(\Delta_{AB}, \eta_{AB})= \frac{W^{\text{sand}}(\alpha^{\star})}{W^{\text{sand}}(\bar{\alpha})}
\end{align*}
Our main theorem bounds the price of anarchy as a function of the curvature and slippage parameters of the network. Notably, given sufficient liquidity conditions on the CFMM network, the price of anarchy is upper bounded by a \textit{constant}. We show this result by connecting the output functions of the CFMMs to $(\lambda, \mu)$-smoothness arguments from~\cite{roughgarden2015intrinsic}. 
\begin{theorem}\label{thm2}
Suppose that $f(\kappa, \mu, \eta_{AB}), g(\kappa, \mu, \eta_{AB}) \in O((1+(\mu  \kappa)^{O(1)})^{1/\mathsf{diam}(G)}$. 
Then there exists a function $C(\kappa, \alpha, \beta, \mu, \eta)$ that is constant in the size of the network graph $G$ such that 
   \begin{align}
        \mathsf{PoA}(\Delta_{AB}, \eta_{AB}) \leq C(\kappa, \alpha, \beta, \mu, \eta)
   \end{align}
\end{theorem}
The proof of this theorem relies on the following lemma: 
\begin{lemma}\label{lemma: PoALemma}
    For every $\alpha, \alpha' \in \hat{S}_{|\mathcal{P}}$, we have: 
    \begin{align*}
        \frac{G_p^{\text{sand}}(\alpha)}{\alpha_p} \alpha_p' \geq \lambda G_p^{\text{sand}}(\alpha) + \nu G_p^{\text{sand}}(\alpha')
    \end{align*}
    for $ \lambda = \nu = \frac{\left( \mu + \mu f(\kappa, \mu, \eta_{AB})^{|p|} - \kappa g(\kappa, \mu, \eta_{AB})^{|p|} \right)}{\left( \kappa + \kappa g(\kappa, \mu, \eta_{AB})^{|p|} - \mu f(\kappa, \mu, \eta_{AB})^{|p|} \right) } $.
\end{lemma}
\begin{proof}
    We can directly utilize the local quadratic bounds implied above to get:     
    \begin{align*}
    -\frac{{G}_p^{\text{sand}}(\alpha)}{\alpha_p} \alpha_p' & \leq -  \left( \kappa + \kappa g(\kappa, \mu, \eta_{AB})^{|p|} - \mu f(\kappa, \mu, \eta_{AB})^{|p|} \right) \frac{\alpha_p^2}{\alpha_p} \alpha_p' \\ & \leq - \left( \kappa + \kappa g(\kappa, \mu, \eta_{AB})^{|p|} - \mu f(\kappa, \mu, \eta_{AB})^{|p|} \right) (\alpha_p^2 + \alpha_p'^2) \\ & \leq - \frac{\left( \kappa + \kappa g(\kappa, \mu, \eta_{AB})^{|p|} - \mu f(\kappa, \mu, \eta_{AB})^{|p|} \right) }{\left( \mu + \mu f(\kappa, \mu, \eta_{AB})^{|p|} - \kappa g(\kappa, \mu, \eta_{AB})^{|p|} \right)} (G_p^{\text{sand}}(\alpha) + G_{p}^{\text{sand}}(\alpha'))
    \end{align*}
    Negating both sides of this inequality, we have the desired result for \begin{align*}
        \lambda = \nu =  \frac{\left( \kappa + \kappa g(\kappa, \mu, \eta_{AB})^{|p|} - \mu f(\kappa, \mu, \eta_{AB})^{|p|} \right) }{\left( \mu + \mu f(\kappa, \mu, \eta_{AB})^{|p|} - \kappa g(\kappa, \mu, \eta_{AB})^{|p|} \right)}
    \end{align*}
\end{proof}

Given, this lemma, we have the following proof of Theorem \ref{thm2}.
\begin{proof}
    We see that we have the following chain of inequalities: 
    \begin{align*}
        W^{\text{sand}}(\bar{\alpha}) & \geq \sum_{p \in \mathcal{P}} \frac{G_p^{\text{sand}}(\alpha^*)}{\alpha_p^*} \bar{\alpha}_p \\ & \geq \sum_{p \in \mathcal{P}} \lambda G_p^{\text{sand}}(\alpha^*) + \nu G_p^{\text{sand}} (\bar{\alpha}) \\ & = \lambda W^{\text{sand}}(\alpha^*) + \nu W^{\text{sand}}(\bar{\alpha})
    \end{align*}
    where the first inequality is by the definition of equilibrium and the second by Lemma \ref{lemma: PoALemma}. Moving the term $\nu W^{\text{sand}}(\bar{\alpha})$ over to the left hand side, and dividing by $W^{\text{sand}}(\bar{\alpha})$ on both sides, we have the desired result:
    \begin{align*}
       \textsf{PoA}(\Delta_{AB}, \eta_{AB}) =  \frac{W^{\text{sand}}(\alpha^*)}{W^{\text{sand}}(\bar{\alpha})} \leq \frac{1-\lambda}{\nu}
    \end{align*}
\end{proof}
\noindent One can view the condition of bounds on $f$ and $g$ informally as saying that provided there is enough liquidity on each edge in the graph (measured by the $\mu, \kappa$ dependence in $f, g$), then the PoA from sandwiching is constant.
We can interpret this result as a weak generalization of the Braess example; it demonstrates that sandwiches do not necessarily cause asymptotically worse performance in routing through CFMM networks, provided there is enough liquidity.
We note that our result can likely be sharpened and that the constants are not tight.

\section{Reordering MEV}\label{sec:reordering}
We now introduce the notion of reordering MEV for sandwich attacks. Throughout this section, we deal with a block of size $3n$, where $n$ is the number of user trades, and a sequence of trades $T_n \coloneqq  \{(\Delta_1, \eta_1), \dots, (\Delta_k, \eta_n)\}$ with a default ordering $i= 1,\dots, n$. (This block size ensures that there are enough slots in the block for a sandwich attacker to insert trades to sandwich all $n$ user trades).
We assume that the slippage limits $\eta_i$ are lower bounded by a single $\eta$, \ie, $\eta_i \geq \eta$ for all $i = 1, \dots, n$.
\\ \\ 
We aim to understand how much sandwiching can worsen user trade execution price by reordering these user trades by a permutation $\pi \in S_n$. This execution price impact is captured by quantity we call the \emph{cost of feudalism}\footnote{
    The term `cost of feudalism' analogizes the structure of feudal kingdoms to miners and searchers, who exact a tax from users in the form of sandwich attacks.
}, or CoF:  
\begin{align}\label{eq: costoffeudalism}
  \mathsf{CoF}(T_n) = \frac{\Expect_{\pi \sim S_n} \left[ \max_{i \in [n]} \lvert \mathsf{PNL}_{\pi(i)}(T_n) - \mathsf{PNL}_{i}(T_n) \rvert \right]}{\Expect_{\pi \sim S_n} \left[ \frac{1}{n} \sum_{i=1}^{n} | \mathsf{PNL}_{\pi(i)}(T_n) - \mathsf{PNL}_{i}(T_n)| \right] },
\end{align}
where $\mathsf{PNL}_i(T_n)=  \ds'_i - \ds_i$ for the $i$th trade in $T_n$ and $\mathsf{PNL}_{\pi(i)}(T_n)$ is the $\pi(i)$th trade in $T_n$ for a permutation $\pi \in S_n$. We will slightly abuse notation and elide the explicit mention of $T_n$ by writing $\mathsf{PNL}_i$ for brevity.
\\ \\ 
\noindent This quantity compares the profit captured from sandwiching the worst affected user to the profit from sandwiching the average user, over all permutations of the trades, relative to a fixed ordering.
Therefore, the CoF characterizes the maximum amount that any individual user's price execution might be affected over reorderings.
We seek to \textit{upper bound} the numerator and \textit{lower bound} the denominator of \eqref{eq: costoffeudalism} to get a bound for $\mathsf{CoF}(T_n)$. Our main result (informally) shows the following:

\begin{center}
    \textit{Given sufficient locality and liquidity conditions on $\mu, \kappa, \alpha, \beta, \eta$, $\mathsf{CoF}(T_n)$ is $O(\log n)$}.
\end{center}

\noindent Now, for a fixed ordering $i = 1, \dots, n$, we show useful bounds on $\ds_i$ and $\ds_i'$ in this sequence, using the upper bound derived in the previous section. To lighten notation, define $\xi_j = \ds_j + \Delta_j - \ds_j'$.
We know that $\ds_i$ and $\ds'_i$ satisfy the equations:
\begin{align*}
  G \left( \ds_i + \Delta_i + \sum_{j=1}^{i-1} \xi_j \right)  - G \left( \ds_i + \sum_{j=1}^{i-1} \xi_j \right) = (1-\eta) G(\Delta_i ) \\
  \ds'_i =  \ds_i' + \de - G^{-1} ( G(\ds_i + \Delta_i) - G(\ds_i)) 
\end{align*} 
Note that for a fixed ordering, the sandwich attacks $\ds_i$ and $\ds_i'$ depend on both the trades $\de_j$ and the sandwich attacks $\ds_j$ and $\ds_j'$ that came \textit{before} (for $j = 1,\dots, i-1$).
This dependence requires us to bound the sandwich attack on trade $i$ in terms of the partial trade drifts up to that trade, which are defined as follows:
\begin{defn}
Define the \emph{partial trade drifts} as $\tilde{u}_i = \sum_{j=1}^{i} \xi_j$.
\end{defn}
\begin{note}
In the subsequent, we always work in the regime $\tilde{u}_i \approx 0$. That is, the trades are roughly mean-reverting. 
Our results hold when this is not true, but in that case there are a number of higher-order correction terms that cloud the intuition of the main results.
\end{note}
We next upper and lower bound $\mathsf{PNL}_i$ in terms of $\tilde{u}_i$ and $\de_i$ using the following proposition:
\begin{prop}\label{prop1}
If a sequence of trades $T_n = \{(\Delta_1, \eta_1), \ldots, (\Delta_n, \eta_n)\}$ is strongly local, then there exist constants $d,e > 0$, $e < 1$ which only depend on $\mu, \kappa, \eta = \max_i \eta_i$, such that:
\begin{align}
    \ds_i \leq (1+d) \de_i + \sum_{j=1}^{i-1} (1+e)^{i-j} \xi_j.
\end{align}
\end{prop}
We prove this proposition using a series of lemmas that are stated and proved in Appendix~\ref{app:prop1Lemmas}.
This bound allows us to suitably `localize' sandwich attacks by upper bounding $\ds_i$ by a term \textit{linear} in the trade $\de_i$ and terms geometrically decaying in the drifts $\tilde{u}_j$ for $j = 1,\dots, i-1$.
This will allow us to bound $\mathsf{CoF}(T_n)$ as a function of the curvature constants $\mu, \kappa, \alpha, \beta$, the slippage limit $\eta$, and the trade drifts $\tilde{u}_i$. 

To bound $\mathsf{CoF}(T_n)$, we first show the sandwich profit function $\mathsf{PNL}_i = \ds_i' - \ds_i$ has curvature dependent on $\mu, \kappa, \alpha, \beta, \eta, \tilde{u}_i$:
\begin{prop}\label{prop:const_curvature}
If a set of trades $T = \{(\Delta_1, \eta_1), \ldots, (\Delta_n, \eta_n)\}$ is strongly local, there exist polynomials $p, q$ of constant degree $d = O(1)$, such that $\mathsf{PNL}_i$ satisfies
\[
q(\mu, \kappa, \alpha, \beta, \eta, \tilde{u}_i)\Delta_i \leq \mathsf{PNL}_i \leq p(\mu, \kappa, \alpha, \beta, \eta, \tilde{u}_i)\Delta_i.
\]
\end{prop}
\noindent We prove this result in Appendix~\ref{app:prop3}.
We combine Proposition \ref{prop:const_curvature} with \cite[Thm. 1]{chitra2021differential} and Lemma 7 of Appendix~\ref{app:prop1Lemmas} to get the following result (which is proved in Appendix~\ref{app:prop4}):
\begin{prop}\label{propBoundMax}
If a set of trades $T = \{(\Delta_1, \eta_1), \ldots, (\Delta_n, \eta_n)\}$ is strongly local, we have:
\[
\Expect_{\pi \sim S_n} \left[ \max_{i \in [n]} \,\lvert \mathsf{PNL}_{\pi(i)} - \mathsf{PNL}_{i} \rvert \right] = O(\log n),
\]
where the constant depends on $\mu, \kappa, \alpha, \beta, \eta, \tilde{u}_i$.
\end{prop}
\begin{prop}\label{propBoundAvg}
If a set of trades $T = \{ (\de_1, \eta_1) , \dots, (\de_n, \eta_n)\}$ is strongly local, we have:
\begin{align*}
    \Expect_{\pi \sim S_n} \left[ \frac{1}{n} \sum_{i=1}^{n} \lvert \mathsf{PNL}_{\pi(i)} - \mathsf{PNL}_i \rvert \right] = \Omega(1)
\end{align*}
\end{prop}

\noindent Combining Propositions \ref{propBoundMax} and \ref{propBoundAvg}, we have our main theorem:
\begin{theorem}\label{thm1}
If a set of trades $T_n = \{(\Delta_1, \eta_1), \ldots, (\Delta_n, \eta_n)\}$ is strongly local, then $\mathsf{CoF}(T_n) = O(\log n)$.
\end{theorem}
\noindent We note that our result can be viewed as a competitive ratio, or so-called \emph{prophet inequality} bound~\cite{hill1983prophet}, except that we fix the distribution over permutations to be uniform over the entire symmetric group~\cite{arsenis2021constrained}.

\section{Conclusion and Future Work}
In this paper, we provided the first formal description of generic sandwich attacks in arbitrary CFMMs.
Using this description, we explicitly computed bounds on sandwich attack profitability that depend on curvature and liquidity.
These bounds allowed us to analyze two prominent forms of CFMM MEV: reordering and routing MEV.
For reordering MEV, we found, somewhat unexpectedly, that given an order flow of $n$ trades, the worst case price impact received by a user is only logarithmically worse than the average case price impact.
Even more paradoxically, we found that for routing MEV, sandwich attacks can, in certain cases, increase social welfare for users when user trades are selfishly routed across a network of CFMMs.
We generalized this example to larger networks of CFMMs and showed that the price of anarchy for routing MEV is constant given sufficient liquidity on the network graph.
To prove a constant price of anarchy, we adapted CFMM price impact functions in a way such that we could apply $(\lambda, \mu)$-smoothness results from \cite{roughgarden2015intrinsic}.

Prior works on MEV \cite{bartoletti2021maximizing, zhou2021a2mm, babel2021clockwork, qin2021quantifying, heimbach2022eliminating} have focused on illuminating either specific profitable attacks or methodologies for the numerical or empirical estimation of MEV.
We believe that this is the first paper which adds more formal algorithmic game theory and probability results about the impact of MEV on users.
Such results not only provide asymptotic, theoretical insight into the nature of MEV, but also suggest some mitigating strategies for protocol developers.
For instance, Theorem \ref{thm1} (and the upper and lower bounds on sandwich profitability) suggest how to set slippage limits $\eta$ as a function of CFMM curvature to reduce the impact of MEV.
These bounds are relatively weak and can be improved if one specializes them to a smaller set of CFMMs.

From a theoretical perspective, our main results bound competitive ratios.
While we only demonstrated that the Cost of Feudalism is $O(\log n)$ for local trades, it is likely that existing work on competitive ratios can generalize this to longer-ranged (non-local) sequences of transactions.
For instance, prophet inequalities are known for online knapsack problems that are similar to those used in MEV~\cite{jiang2022tight}.
Results on composition of prophet inequalities~\cite{lucier2017economic} may allow us to extend both routing and reordering bounds to non-local forms of MEV.
Another avenue for extending the results of this paper is to map long-ranged MEV strategies to bounded auctions and then utilize the competitive ratio results for these auctions (\eg, the results of~\cite{makhdoumi_auction}).

The remaining set of papers in this series will focus on other types of MEV and the interaction between MEV, privacy, and the allocative efficiency and fairness of ordering transactions.
Many of the results in this paper were inspired by \cite{chitra2021differential}, which made a direct connection between MEV profit versus the cost of privacy.
This connection is most easily analyzed in CFMMs, where the `cost of privacy' can be directly interpreted via the excess price impact or fees that a user has to pay to ensure that MEV searchers have negligible profitability.
Extending this analogy to more complex forms of MEV, such as cross-chain MEV~\cite{obadia2021unity} and MEV related to liquidations~\cite{kaoCompound}, is future work. Liquidations, in particular, are a kind of MEV that is heavily dependent on the ordering of transactions across the block.

\section{Acknowledgments}
We thank Guillermo Angeris and Alex Evans for helpful comments and feedback.

\newpage
\printbibliography

\newpage 
\appendix 
\section{Uniswap Sandwich Example}\label{sec:uniswapAppendix}
The defining relation for $\ds$ for Uniswap is: 
\begin{align*}
    -\frac{k}{R +\de + \ds} + R' + \frac{k}{R + \ds} - R' = (1-\eta) \left( -\frac{k}{R +\Delta} + R' \right)
\end{align*}
Cancelling $R'$ and putting the left and hand sides over a common denominator: 
\begin{align*}
    \frac{-k(R+ \ds) + k(R+ \de + \ds)}{(R+ \de +\ds)(R + \ds)}  = (1-\eta) \left(\frac{-k +R' R + \de R'}{R + \de} \right)
\end{align*}
Noting that $k = R' R$, and simplifying the left hand side: 
\begin{align*}
    \frac{k \Delta}{R^2 + R \de + R \ds + R \ds + \de \ds + \ds^2 } = (1-\eta) \frac{\de R'}{R+\de}
\end{align*}
Which, after cancelling $R'\de$ on both sides, raising both sides of the equation to the $-1$ power, and multiplying by $R$  gives us the equation:
\begin{align*}
    \ds^2 + \ds(\de + 2R) + (R^2 +R \de) = \frac{1}{1-\eta} (R^2 + R\de)
\end{align*}
and moving the right hand side to the left we have:
\begin{align*}
    \ds^2 + \ds(\de + 2R) + (R^2 + R \de)\left(1 - \frac{1}{1-\eta} \right) = 0
\end{align*}
We solve the above quadratic to give us: 
\begin{align*}
    \ds =  \frac{-(\de + 2R) \pm \sqrt{( \de+2R)^2 -4 (R^2 + R \de) \frac{-\eta}{1-\eta}}}{2} 
\end{align*}
Taking the positive root (which is the correct root to take, as when $\eta =0$, the positive root gives us $\ds = 0$, and also gives us that $\ds$ is \textit{increasing} in $\eta$), we have: 
\begin{align*}
    \ds =  \frac{-(\de + 2R) + \sqrt{( \de+2R)^2 -4 (R^2 + R \de) \frac{-\eta}{1-\eta}}}{2} 
\end{align*}

\section{Price Slippage and Quantity Slippage are Equivalent}\label{sec:appendixA}
Uniswap enforces slippage limits~\cite{uniswap_check} in quantity space rather than price space.
In particular, the interface for making a trade takes in an input quantity $\Delta$ and a minimum output quantity $\Delta'_{\min}$ and enforces that the output quantity received for $\Delta$, $\Delta'$ always satisfies $\Delta' \geq \Delta'_{\min}$.
In this section, we show that for Uniswap, there is a way to map slippage limits defined in terms of quantity to those defined in terms of price.
For general CFMMs, this is also possible to show, but it is quite impractical to perform in practice as it involves inverting the CFMM invariant function.

Recall that the output quantity for a trade of size $\Delta$, is defined via the forward exchange function $G(\Delta)$~\cite[\S4.1]{angeris2021cfmm}:
\[
\Delta' = G(\Delta) = \int_{0}^{-\Delta} g(t) dt
\]
Enforcing the Uniswap condition on quantity is equivalent to
\begin{equation}\label{eq:qty_space_slip}
G(\Delta + \ds) -G(\ds) \geq (1-\eta^q)G(\Delta) 
\end{equation}
for any front running trade $\ds$. 
The left hand side of this equation is the output quantity received if one submits a trade of size $\Delta$ after someone has submitted a trade of size $\ds$ ahead of the user's trade.
The right hand side provides the notion of minimum quantity with $\eta^q$ defined as a quantity space slippage limit.
Put together, equation \eqref{eq:qty_space_slip} states that the output quantity must at least be $(1-\eta^q)$ times the quantity that was expected assuming no front running transactions take place.
Our claim in this section is that one can compute $\eta^p(\eta^q)$ to go between quantity space slippage limits and price space slippage limits.

Recall that there exists $\mu, \kappa > 0$ such that the forward transfer function satisfies $G(\Delta) \geq \kappa \Delta$ and $G(\Delta) \leq \mu \Delta^2$.
This implies that
\begin{equation}\label{eq:unilb}
\Delta' = G(\Delta + \ds) - G(\ds) \geq (1-\eta^q)\kappa \Delta
\end{equation}
Suppose that the Uniswap pool's initial reserves are $(R, R')$ with an initial spot price of $p_0 = \frac{R'}{R}$
Then we can write a price condition akin to that of \S\ref{s-cfmm} as
\begin{align*}
g(\Delta) - g(0) &= \frac{R'-\Delta'}{R+\Delta} - \frac{R'}{R} = \frac{p_0 R - \Delta'}{R + \Delta} \\
&\leq \frac{p_0 R + (\eta^q - 1)\kappa \Delta}{R+\Delta} = \frac{p_0 R - \kappa \Delta}{R+\Delta} + \frac{\eta^q \kappa\Delta}{R+\Delta} 
\end{align*}
where the inequality uses \eqref{eq:unilb}.
Now we have
\[
\frac{g(\Delta)-g(0)}{g(\Delta)} = \frac{p_0 R - \kappa \Delta}{R'-\Delta'} + \frac{\eta^q \kappa\Delta}{R'-\Delta'} \leq \frac{p_0R - \mu\Delta}{R'- (1 - \eta^q)\mu\Delta} + \frac{\eta^q \mu\Delta}{R' - (\eta^q - 1)\mu\Delta}
\]
The common denominator can be expanded as
\[
\frac{1}{R'-(1-\eta^q)\mu\Delta} = \frac{1}{R'}\sum_{n=0}^{\infty} \left(\frac{(1-\eta^q)\mu\Delta^2}{R'}\right)^n \leq \frac{c(1-\eta^q)\mu\Delta^2}{R'^2}
\]
By setting $\eta$ to this quantity we match the bound from \S\ref{s-cfmm}.

\section{Bounds on $\ds$}\label{app:dsbounds}
\subsection{Upper Bound (Claim \ref{ds_ub_claim})}
Note that by construction, $G(0) = 0$~\cite[\S4.1]{angeris2021cfmm}. These bounds have $\kappa$ and $\mu$ in units of price. That is, the linear lower and upper bounds on the quantity output by a trade of size $\Delta$ implies a maximum and minimum price impact (lower and upper bounds on $g(\cdot)$). 
\\ \\ 
We first bound the price impact of a cumulative trade $(\ds, (\Delta, \eta), \ds')$. To do this, we first assume the sandwich attack is optimal (making the output quantity the user gets tight with the specified slippage): 
\begin{align}\label{eq:capG-slippage}
    G(\ds + \Delta) - G(\ds) = (1-\eta) G(\Delta)
\end{align}
If we force the lower bound implied by \eqref{eq:capG-bd} to be greater than the upper bound in equation \eqref{eq:capG-slippage}, we have
\begin{align*}
    G(\Delta + \ds) - G(\ds) \geq \kappa(\Delta + \ds) - \mu \ds \geq (1-\eta)\mu \Delta \geq (1-\eta) G(\Delta)
\end{align*}
Rearranging the middle two terms gives us:
\begin{align*}
    (\kappa - \mu) \ds &\geq (1-\eta) \mu \Delta - \kappa \Delta = (\mu - \kappa) \Delta - \eta \mu \Delta 
\end{align*}
and dividing by $\kappa - \mu \leq 0$ we have: 
\begin{equation}\label{eq:capG-ds-ub}
\ds \leq \left(\frac{\eta\mu}{\mu-\kappa} - 1\right) \Delta
\end{equation}
which provides an upper bound on $\ds$ as a function of the slippage and curvature.
The bracketed term is positive when $\eta \geq 1 - \frac{\kappa}{\mu}$ which means that the slippage limit set by the user is larger than inverse of the curvature ratio.
The smaller $\eta$ is, i.e. the smaller a discount the user is willing to accept on the minimum output quantity they receive, the smaller the upper bound on the sandwich size will be. 

\subsection{Lower Bound (Claim \ref{ds_lb_claim})}
Suppose we have a $(\mu, \kappa)$-smooth forward exchange function $G(\Delta)$ whose derivative $g(\delta) = G'(\Delta)$ is $\beta$-liquid, \eg~$g(-\delta) - g(0) \geq \beta \Delta$.
To construct a lower bound on $\ds$, we will bound the left side of \eqref{eq:sandwich-defn} below and the right side above.
We construct a quadratic lower bound for $G(\Delta)$ using $g$:
\[
G(\Delta) = \int_{0}^{\Delta} g(-t) dt \geq \int_{0}^{\Delta} \beta t + g(0) dt = \frac{\kappa\Delta^2}{2} + g(0)\Delta
\]
Using this bound, we can lower bound the left side of \eqref{eq:sandwich-defn} as
\[
G(\ds + \Delta) - G(\ds) \geq \frac{\beta(\Delta+\ds)^2}{2} + g(0) \Delta - \mu\ds
\]
Combining this with the bound $(1+\eta) G(\Delta) \leq (1+\eta)\mu \Delta$ gives the condition
\[
(1+\eta)\mu \Delta \leq  \frac{\beta(\Delta+\ds)^2}{2} + g(0) \Delta - \mu\ds
\]
Solving this quadratic equation in $\ds$ for when there is equality yields two roots
\[
r_{\pm} = \left(\frac{\mu}{\beta} - \Delta\right)\left[1 \pm \sqrt{1 + \beta\Delta\left(1+\frac{\eta\mu+p_0}{\mu-\kappa\Delta}\right)}\right]
\]
Provided that $\Delta < \frac{\mu}{\beta}$ (note that $\mu$ has units of price where $\beta$ has units of price over quantity), then $r_+ > 0$ and we have the condition
\begin{align}\label{eq:ds-lb}
\ds \geq r_+ &= \left(\frac{\mu}{\beta} - \Delta\right)\left[1 + \sqrt{1 + \beta\Delta\left(1+\frac{\eta\mu+p_0}{\mu-\kappa\Delta}\right)}\right] > \left(\frac{\mu}{\beta} - \Delta\right)\gamma
\end{align}
where $\gamma > 1$.

\section{Bounds for $\ds'$}\label{app:dsprime}

Once again using the linear lower and upper bounds on $G(\cdot)$ (and $G^{-1}(\cdot)$), we have:
\begin{align}\label{eq:ginv_ub}
    G^{-1}(G(\ds + \Delta) - G(\ds)) &\leq \frac{1}{\kappa} \left( \mu (\ds + \Delta) - \kappa \ds \right) = \frac{1}{\kappa}\left((\mu-\kappa)\ds + \mu \Delta\right) \nonumber\\
    &\leq \frac{1}{\kappa}\left((\mu-\kappa) \left(\frac{\eta\mu}{\mu-\kappa} - 1\right) + \mu\right) \Delta = \left(\frac{\eta\mu}{\kappa} + 1\right)\Delta
\end{align}
Similarly, we have a matching lower bound
\begin{align}\label{eq:ginv_lb}
    G^{-1}(G(\ds + \Delta) - G(\ds)) &\geq \frac{1}{\mu} \left( \kappa (\ds + \Delta) - \mu \ds \right) = \frac{1}{\mu}\left(-(\mu-\kappa)\ds + \mu \Delta\right) \nonumber \\
    &\geq \frac{1}{\mu}\left((\mu-\kappa)\left(1-\frac{\eta\mu}{\mu-\kappa}\right)\Delta + \mu \Delta\right) \nonumber \\
    &= \left(\left(1-\frac{\kappa}{\mu}\right) + (1-\eta)\right)\Delta
\end{align}
Given this setup we have the following proofs of Claims \ref{claim:dsprime-ub} - \ref{claim:hurdle_rate}:
\subsection{Proof of Claim \ref{claim:dsprime-ub}}
These bounds furnish us bounds for $\ds'$:
\begin{align*}
    \ds' &= \ds + \Delta -  G^{-1}( G(\ds + \Delta) - G(\ds)) \\
    &\leq \frac{\eta\mu}{\mu-\kappa}\Delta - \left(\left(1-\frac{\kappa}{\mu}\right) + (1-\eta)\right)\Delta \\
    &= \left(\eta\left(1 + \frac{\mu}{\mu-\kappa}\right) - \left(2-\frac{\kappa}{\mu}\right)\right)\Delta
\end{align*}
\subsection{Proof of Claim \ref{claim:dsprime-lb}}
\begin{align}\label{eq:dsprime-lb}
    \ds' &= \ds + \Delta -  G^{-1}( G(\ds + \Delta) - G(\ds)) \nonumber \\
         &\geq \left(\frac{\mu}{\beta}-\Delta\right)\gamma + \Delta - \left(\frac{\eta\mu}{\kappa} + 1\right)\Delta \nonumber \\
         &\geq \frac{\mu\gamma}{\beta} - \Delta\left(\gamma + \frac{\eta\mu}{\kappa}\right)
\end{align}
\subsection{Proof of Claim \ref{claim:hurdle_rate}}

The proof of claim \ref{claim:hurdle_rate} can be see simply by using the inequality
\[
\ds' - \ds \geq \left(\eta\left(1 + \frac{\mu}{\mu-\kappa}\right) - \left(2-\frac{\kappa}{\mu}\right)\right)\Delta + \Delta \gamma - \frac{\mu\gamma}{\beta}
\]

\section{Bounds for $\mathsf{CoF}(T_n)$}\label{app:prop1Lemmas}
\subsection{Statements of Lemmas}
The first two lemmas provide upper and lower bounds on $\ds_i$ using the partial trade drifts $\tilde{u}_{i-1}$ and the trade $\de_i$:
\begin{lemma}\label{lemma1}
\begin{align*}
     \ds_i \leq \tilde{u}_{i-1}  + \left( \frac{\eta \mu}{\mu - \kappa} - 1\right) \de_i
\end{align*}
\end{lemma}

\begin{lemma}\label{lemma2}
    \begin{align*}
        \ds_i & \geq \tilde{u}_{i-1} + \left( \frac{\eta \kappa}{\mu - \kappa}  -1 \right) \de_i
    \end{align*}
\end{lemma}

The next two lemmas provide upper and lower bounds on $\ds_i'$ in terms of $\tilde{u}_{i-1}$ and $\de_i$. We note here that these bounds differ from the linear upper bounds in Lemmas~\ref{lemma1} and $\ref{lemma2}$ as they use a quadratic bound on $\ds_i$ to bound $\ds_i'$. These stronger conditions are necessary to prove Proposition~\ref{prop1}. We also note that the constants $\gamma$ and $\nu$ in the lemmas below are related to solutions of the aforementioned quadratic equation. In particular, $\nu < 0$ and $\gamma > 0$.
\begin{lemma}\label{lemma3}
\begin{align*}
    \ds'_i \geq -\left(\frac{\eta \mu}{\mu - \kappa} + \frac{\kappa}{\mu}(\gamma+1) - 1 \right) \de_i + \frac{\kappa}{\mu} \left(\frac{\mu}{\beta} - g(0) \right)  - \left(1 + \frac{\kappa \gamma}{\mu}\right) \tilde{u}_{i-1}
\end{align*}
\end{lemma}

\begin{lemma}\label{lemma6}
\begin{align*}
    \ds'_i \leq - \left( \frac{\eta \kappa}{\mu - \kappa} + \frac{\mu}{\kappa} (\nu + 1) - 1\right) \de_i + \frac{\mu}{\kappa} \left(\frac{\kappa}{\beta} - g(0) \right) - \left(1 + \frac{\mu \nu}{\kappa} \right)\tilde{u}_{i-1}
\end{align*}
\end{lemma}
We now use these lemmas to derive lower and upper bounds on $\mathsf{PNL}_i = \ds_i' - \ds_i$. We first recall discrete Gr\"onwall inequalities from~\cite{agarwal2000difference}:
\begin{prop}[Thm 4.1.1,~\cite{agarwal2000difference}]\label{prop2}
Suppose that $u_k, q_k, f_k \in \reals$ are non-negative sequences and $p_k \in \reals$ is a sequence that collectively satisfy:
\[
u_k \leq p_k + q_k \sum_{\ell=a}^{k-1}f_{\ell}u_{\ell}
\]
Then for all $k \geq a$, we have:
\[
u_k \leq p_k + q_k \sum_{\ell = a}^{k-1}p_{\ell} f_{\ell} \left(\prod_{i=\ell+1}^{k-1} (1+q_i f_i)\right)
\]
\end{prop}

\begin{prop}[Thm 4.1.9,~\cite{agarwal2000difference}]\label{prop3}
Let for all $k,r \in \mathbb{N}$ such that $k \leq r$ the following inequality be satisied: 
\begin{align*}
    u_r \geq u_k - q_r \sum_{\ell = k+1}^{r} f_{\ell} u_{\ell}
\end{align*}
where $u_k$ is not necessarily nonnegative. Then, for all $k,r \in \mathbb{N}$, $k \leq r$
\begin{align*}
    u_r \geq u_k \prod_{\ell = k+1}^{r} (1 + q_r f_{\ell} )^{-1}
\end{align*}
\end{prop}

We now provide upper and lower bounds on $\mathsf{PNL}_i = \ds_i' - \ds_i$, which is the sandwich profit extracted by the sandwich attacker over the $i$-th trade, $\de_i$. We make use of the above Grönwall inequalities to unroll the recursion and get bounds on $\mathsf{PNL}_i$ only in terms of $\de_j$ for $j = 1,\dots, i$. Combining the bounds from Lemmas~\ref{lemma6} and~\ref{lemma2} we have the following upper bound:
\begin{lemma}\label{lemma7}
We can upper bound  $\ds_i'-\de_i - \ds_i$ as:
\begin{align*}
       \ds_i'- \de_i - \ds_i & \leq \left( -1 - \frac{\mu}{\kappa} (\nu + 1 ) \right) \de_i + \frac{\mu}{\kappa} \left(\frac{\kappa}{\beta} - g(0) \right) \\ & + \left(2 + \frac{\mu \nu}{\kappa} \right) \left(\sum_{j=1}^{i-1} \ds_i' - \de_i - \ds_i\right)
\end{align*}
\end{lemma}
Now, using Proposition~\ref{prop2} we have:
\begin{lemma}\label{lemmaPNLUB}
For $p_i =  \left( -1 -\frac{\mu}{\kappa} (\nu + 1) \right) \de_i + \frac{\mu}{\kappa} \left(\frac{\kappa}{\beta} - g(0) \right)$, the sandwich profit $\mathsf{PNL}_i = \ds'_i - \ds_i$ can be upper bounded: 
\begin{align*}
    \ds_i' - \ds_i \leq p_i + \de_i + \left( 2 +\frac{\mu \nu}{\kappa} \right) \sum_{\ell = 1}^{i-1} p_{\ell} \left( 3 + \frac{\mu \nu}{\kappa} \right)^{i-\ell - 1}
\end{align*}
\end{lemma}
We have similar lower bounds for $\mathsf{PNL}_i$ using the following lemmas:

\begin{lemma}\label{lemma8}
We can lower bound $\ds_i' - \de_i - \ds_i$ as: 
\begin{align}\label{eq:gron}
    \ds_i' - \de_i - \ds_i & \geq  \left( - 1 +\frac{\kappa}{\mu}(\gamma + 1) \right) \de_i  + \frac{\kappa}{\mu}\left( \frac{\mu}{\beta} - g(0) \right) \\ & + \left( 2 + \frac{\kappa \gamma}{\mu} \right) \left( \sum_{j=1}^{i-1} \ds_j' - \de_j - \ds_j\right)
\end{align}
\end{lemma}

\subsection{Proofs of Lemmas}
\begin{lemma}
\begin{align*}
    \ds_i \leq \tilde{u}_{i-1}  + \left( \frac{\eta \mu}{\mu - \kappa} - 1\right) \de_i
\end{align*}
\end{lemma}
\begin{proof}
We first note that by definition, $\ds_i$ satisfies the equation:
\begin{align*}
    & G\left( \ds_i + \de_i + \sum_{j=1}^{i-1} \ds_i + \de_i - \ds'_i \right) \\ &- G \left(\ds_i + \sum_{j=1}^{i-1} \ds_i + \de_i - \ds'_i \right) = (1-\eta) G(\de_i)
\end{align*}
Using the bounds $G(\de_i) \leq \mu \de_i$ and $G(\de_i) \geq \kappa \de_i$, we lower bound the left hand side and upper bound the right hand side to get: 
\begin{align*}
    \left(\kappa - \mu\right) \ds_i + \kappa \de_i + (\kappa - \mu) \left( \sum_{j=1}^{i-1} \ds_j + \de_j - \ds'_j \right) \geq (1-\eta ) \mu \de_i
\end{align*}
Rearranging and dividing by $\kappa - \mu$, we have: 
\begin{align*}
    \ds_i & \leq \sum_{j=1}^{i-1} \ds_j + \de_j - \ds'_j + \left(1 - \frac{\eta \mu}{\mu - \kappa } \right) \de_i \\ & = \tilde{\mu}_{i-1} + \left(  \frac{\eta \mu}{\mu - \kappa}  - 1\right) \de_i 
\end{align*}
\end{proof}
\begin{lemma}
    \begin{align*}
        \ds_i & \geq  \tilde{u}_{i-1} + \left(1- \frac{\eta \kappa}{\mu - \kappa}  \right) \de_i
    \end{align*}
\end{lemma}
\begin{proof}
Once again, we begin with the definition of $\ds_i$: 
\begin{align*}
    & G\left( \ds_i + \de_i + \sum_{j=1}^{i-1} \ds_i + \de_i - \ds'_i \right) \\ &- G \left(\ds_i + \sum_{j=1}^{i-1} \ds_i + \de_i - \ds'_i \right) = (1-\eta) G(\de_i)
\end{align*}
We now upper bound the left hand side and lower bound the right hand side to get:
\begin{align*}
    \left(\mu-  \kappa \right) \ds_i + \mu \de_i + (\mu - \kappa) \left(\sum_{j=1}^{i-1} \ds_j + \de_j - \ds'_j \right) \geq (1-\eta) \kappa \de_i
\end{align*}
Rearranging and dividing by $\mu - \kappa$ we have: 
\begin{align*}
    \ds_i & \geq \sum_{j=1}^{i-1} \ds_j + \de_j - \ds'_j + \left(  1-\frac{\eta \kappa}{\mu - \kappa}   \right) \de_i \\ & = \tilde{u}_{i-1} + \left(1- \frac{\eta \kappa}{\mu - \kappa}  \right) \de_i
\end{align*}
\end{proof}
\begin{lemma}
\begin{align*}
    \ds'_i \geq -\left(\frac{\eta \mu}{\mu - \kappa} + \frac{\kappa}{\mu}(\gamma+1) - 1 \right) \de_i + \frac{\kappa}{\mu} \left(\frac{\mu}{\beta} - g(0) \right)  - \left(1 + \frac{\kappa \gamma}{\mu}\right) \tilde{u}_{i-1}
\end{align*}
\end{lemma}
\begin{proof}
Note that by definition, $\ds'_i$ satisfies the equation:
\begin{align*}
    \ds'_i & = \ds + \de - G^{-1} \left( G \left( \ds_i + \de_i + \sum_{j=1}^{i-1} \ds_j + \de_j  - \ds'_j \right)\right. \\ 
    & \left. -   G\left( \ds_i + \sum_{j=1}^{i-1} \ds_j + \de_j  - \ds'_j\right) \right) 
\end{align*}
We first get a quadratic lower bound for $\ds_i$ in $\de_i$ and use it to lower bound $\ds'_i$:
\begin{align*}
    (1 +\eta)  \mu \de_i & \leq \frac{\beta (\ds_i + \de_i + \sum_{j=1}^{i-1} \ds_j + \de_j + \ds'_j)^2}{2} \\ &+ g(0) \left(\ds_i + \de_i + \sum_{j=1}^{i-1} \ds_j + \de_j + \ds'_j \right) \\ & - \mu \left( \ds_i + \sum_{j=1}^{i-1} \ds_j + \de_j - \ds'_j \right)
\end{align*}
Next, we solve this quadratic equation in $\ds_i$ for when there is equality. In particular, we solve: 
\begin{align*}
    0 &= \underbrace{\frac{\beta}{2}}_{a} \ds_i^2 +  \underbrace{(\beta \left( \frac{\de_i + \tilde{u}_{i-1}  + g(0)}{2}\right) - \mu)}_{b} \ds_i \\ & +  \underbrace{\frac{\beta}{2} (\de_i + \tilde{u}_{i-1} )^2 + (g(0) - (1+\eta) \mu) \de_i + (g(0) - \mu) \tilde{u}_{i-1} }_{c}
\end{align*}
which gives us the roots: 
\begin{align*}
    r_{\pm} = \left( \frac{\mu}{\beta} - \de_i - \tilde{u}_{i-1} - g(0) \right) \pm \sqrt{ (\de_i + \tilde{u}_{i-1} + g(0) - \mu)^2 - 2 \beta c}
\end{align*}
We can now take the positive root, $r_{+} > 0$ and we have the condition:
\begin{align*}
    \ds_i \geq r_{+} > \left( \frac{\mu}{\beta} - \de_i - \tilde{u}_{i-1} - g(0) \right) \gamma
\end{align*}
for $\gamma > 0$. We can now use the definition of $G^{-1}$ to construct a lower bound for $\ds'_i$: 
\begin{align*}
    \ds'_i & \geq \ds_i + \de_i - \frac{1}{\kappa}  \left( \mu \left( \ds_i + \de_i + \tilde{u}_{i-1} \right) - \kappa \left( \ds_i + \tilde{u}_{i-1} \right)\right) \\ 
    & \geq \ds_i + \de_i  - \frac{1}{\kappa} \left( \mu  \left( \frac{\mu}{\beta} - \de_i - \tilde{u}_{i-1} - g(0) \right) \gamma - \kappa \left( \tilde{u}_{i-1} + \left(1- \frac{\eta \kappa}{\mu- \kappa}\right) \de_i \right) + \kappa \de_i \right) \\ 
    & = \ds_i -\frac{1}{\kappa} \left(  \left(-\kappa + \kappa \left(1- \frac{\eta \mu}{\mu - \kappa} \right) \right) \de_i - \mu (\gamma+1) \de_i   + \kappa \left(\frac{\mu}{\beta} - g(0) \right)  - (\kappa + \mu) \tilde{u}_{i-1} \right) \\
    &= \ds_i +\left(2 - \frac{\eta \mu}{\mu - \kappa} + \frac{\mu}{\kappa}(\gamma+1)  \right) \de_i + \frac{\kappa}{\mu} \left(\frac{\mu}{\beta} - g(0) \right)  - \left(1 + \frac{\mu }{\kappa}\right) \tilde{u}_{i-1}
\end{align*}
\end{proof}
\begin{lemma}\label{lemma4}
\begin{align}
    \ds'_i &\leq \frac{2 (\mu- \kappa)}{\kappa} \tilde{u}_{i-1} + \left(\frac{\mu- \kappa}{\kappa} + \frac{(1-\eta) \mu}{\kappa} \right) \de_i 
\end{align}
\end{lemma}
\begin{proof}
Once again, we begin with the defining relation of $\ds_i$. That is:
\begin{align*}
    \ds'_i & = \ds_i + \de_i - G^{-1} \left( G \left( \ds_i + \de_i + \sum_{j=1}^{i-1} \ds_j + \de_j  - \ds'_j \right)\right. \\ 
    & \left.-   G\left( \ds_i + \sum_{j=1}^{i-1} \ds_j + \de_j  - \ds'_j\right)\right)
\end{align*}
Now, we use the linear upper bound on $\ds_i$ from Lemma~\ref{lemma1} and the curvature of $G$ and $G^{-1}$ to upper bound $\ds_i'$ as:
\begin{align*}
    \ds'_i & \leq \ds_i + \de_i - \frac{1}{\mu} \left( (\kappa - \mu) \ds_i + \kappa \de_i + (\kappa - \mu) \tilde{u}_{i-1}\right) \\ & \leq \ds_i + \de_i -\frac{\kappa - \mu}{\mu} \left( \tilde{u}_{i-1} + \left( \frac{\eta\mu}{\mu - \kappa} -1 \right) \de_i \right) + \frac{\mu}{\kappa} \de_i + \frac{\mu - \kappa}{\kappa} \tilde{u}_{i-1} \\ & \leq \left(1 - \frac{\kappa - \mu}{\mu} \right) \left(\frac{\eta \mu}{\mu - \kappa} - 1 \right) \de_i + \left( \frac{\mu}{\kappa} + 1 \right) \de_i +  \frac{2(\mu - \kappa)}{\mu} \tilde{u}_{i-1}   \\ & = \frac{2 (\mu- \kappa)}{\kappa} \tilde{u}_{i-1} + \left( \left(1 - \frac{\kappa - \mu}{\mu} \right) \left(\frac{\eta \mu}{\mu - \kappa} - 1 \right)  + \frac{\mu}{\kappa} + 1 \right) \de_i 
\end{align*}
\end{proof}
We also prove a quadratic upper bound on $\ds'_i$: 
\begin{lemma}
\begin{align*}
    \ds'_i \leq - \left( \frac{\eta \kappa}{\mu - \kappa} + \frac{\mu}{\kappa} (\nu + 1) - 1\right) \de_i + \frac{\mu}{\kappa} \left(\frac{\kappa}{\beta} - g(0) \right) - \left(1 + \frac{\mu \nu}{\kappa} \right)\tilde{u}_{i-1}
\end{align*}
\end{lemma}
\begin{proof}
\begin{align*}
    \ds'_i & = \ds_i + \de_i - G^{-1} \left( G \left( \ds_i + \de_i + \sum_{j=1}^{i-1} \ds_j + \de_j  - \ds'_j \right)\right. \\ 
    & \left. -   G\left( \ds_i + \sum_{j=1}^{i-1} \ds_j + \de_j  - \ds'_j\right)\right)
\end{align*}
We have: 
\begin{align*} 
    (1+\eta) \kappa \de_i & \geq \frac{\beta (\ds_i + \de_i + \tilde{u}_{i-1})^2}{2}\\ &  + g(0) \left(\ds_i + \de_i + \tilde{u}_{i-1} \right) \\ & - \mu \left( \ds_i + \tilde{u}_{i-1} \right)
\end{align*}
We now solve the quadratic equation:
\begin{align*}
    0 & = \underbrace{\frac{\beta}{2}}_{a}\ds_i^2 \\ & + \underbrace{\left(\beta \left( \frac{\de_i + \tilde{u}_{i-1} + g(0)}{2}\right) \right)}_{b} \ds_i \\& + \underbrace{\frac{\beta}{2} \left( \de_i + \tilde{u}_{i-1} \right)^2 + (g(0) - (1+\eta) \kappa ) \de_i + (g(0) - \kappa) \tilde{u}_{i-1} }_{c}
\end{align*}
Which gives us the roots:
\begin{align*}
    r_{\pm} = \left(\frac{\kappa}{\beta} - \de_i - \tilde{u}_{i-1} - g(0) \right) \pm \sqrt{(\de_i + \tilde{u}_{i-1} + g(0) - \kappa)^2 - 2 \beta c} 
\end{align*}
Therefore, we can take the negative root and upper bound: 
\begin{align*}
    \ds_i \leq \left(\frac{\kappa}{\beta} - \de_i - \tilde{u}_{i-1} - g(0) \right) \nu
\end{align*}
for some $\nu < 0$. We can now use the definition of $G^{-1}$ to construct an upper bound for $\ds'$. 

\begin{align*}
    \ds_i' & \leq \ds_i + \de_i - \frac{1}{\mu} \left(\kappa (\ds_i + \de_i + \tilde{u}_{i-1} ) - \mu (\ds_i + \tilde{u}_{i-1} ) \right) \\
    & \leq \ds_i + \de_i - \frac{1}{\mu} \left(\kappa \left(\frac{\kappa}{\beta} - \de_i - \tilde{u}_{i-1} - g(0) \right) \nu - \mu \left( \tilde{u}_{i-1} + \left(\frac{\eta \mu}{\mu - \kappa} - 1\right) \de_i \right) + \mu \de_i \right) \\ 
    & = \ds_i + \de_i - \frac{1}{\mu} \left( -\kappa \left( \frac{\eta \kappa}{\mu - \kappa}  -1 \right) \de_i - \kappa ( \nu + 1)\de_i + \mu \left( \frac{\kappa}{\beta} - g(0) \right) -(\kappa + \nu \mu) \tilde{u}_{i-1} \right) \\ 
    & = - \left( \frac{\eta \kappa}{\mu - \kappa} + \frac{\mu}{\kappa} (\nu + 1) - 1\right) \de_i + \frac{\mu}{\kappa} \left(\frac{\kappa}{\beta} - g(0) \right) - \left(1 + \frac{\mu \nu}{\kappa} \right)\tilde{u}_{i-1}
\end{align*}

\end{proof}

\begin{lemma}
We can upper bound  $\ds_i' - \de_i - \ds_i$ as:
\begin{align}
       \ds_i'- \de_i - \ds_i & \leq \left( -1 - \frac{\mu}{\kappa} (\nu + 1 ) \right) \de_i + \frac{\mu}{\kappa} \left(\frac{\kappa}{\beta} - g(0) \right) \\ & + \left(2 + \frac{\mu \nu}{\kappa} \right) \left(\sum_{j=1}^{i-1} \ds_i' - \de_i - \ds_i\right)
\end{align}
\end{lemma}
\begin{proof}
We use the lower bound on $\ds$ and upper bound on $\ds_i'$ derived in Lemmas~\ref{lemma2} and $\ref{lemma6}$ respectively to bound $\ds_i' - \ds_i$. We have: 
\begin{align*}
    \ds_i' - \ds_i & \leq - \left( \frac{\eta \kappa}{\mu - \kappa} + \frac{\mu}{\kappa} (\nu + 1) - 1\right) \de_i + \frac{\mu}{\kappa} \left(\frac{\kappa}{\beta} - g(0) \right) - \left(1 + \frac{\mu 
\nu}{\kappa} \right)\tilde{u}_{i-1} \\ & - \tilde{u}_{i-1} + \left( \frac{\eta \kappa}{\mu - \kappa} - 1 \right)\de_i \\ & = \frac{\mu}{\kappa}(\nu + 1) \de_i + \frac{\mu}{\kappa} \left( \frac{\kappa}{\beta} - g(0) \right) - \left( 2 + \frac{\mu \nu}{\kappa} \right) \tilde{u}_{i-1}
\end{align*}
and adding $-\de_i$ to both sides gives us: 
\begin{align*}
     \ds_i' - \de_i - \ds_i & \leq  \left( -1 -\frac{\mu}{\kappa} (\nu + 1) \right) \de_i + \frac{\mu}{\kappa} \left(\frac{\kappa}{\beta} - g(0) \right) - \left(2 + \frac{\mu 
\nu}{\kappa} \right)\tilde{u}_{i-1} \\ & = \left( -1 - \frac{\mu}{\kappa} (\nu + 1 ) \right)\de_i + \frac{\mu}{\kappa} \left(\frac{\kappa}{\beta} - g(0) \right) \\ & + \left(2 + \frac{\mu \nu}{\kappa} \right) \left(\sum_{j=1}^{i-1} \ds_j' - \de_j - \ds_j\right)
\end{align*}
\end{proof}

\begin{lemma}\label{PNLUpperBound}
For $p_i =   \left( -1 -\frac{\mu}{\kappa} (\nu + 1) \right) \de_i + \frac{\mu}{\kappa} \left(\frac{\kappa}{\beta} - g(0) \right)$, the sandwich profit $\mathsf{PNL}_i = \ds'_i - \ds_i$ can be upper bounded: 
\begin{align}
    \ds_i' - \ds_i \leq \left(-\frac{\mu}{\kappa} (\nu + 1) \right) \de_i + \frac{\mu}{\kappa} \left( \frac{\kappa}{\beta} - g(0) \right) + \left( 2 +\frac{\mu \nu}{\kappa} \right) \sum_{\ell = 1}^{i-1} p_{\ell} \left( 3 + \frac{\mu \nu}{\kappa} \right)^{i-\ell - 1}
\end{align}
\end{lemma}

\begin{proof}
We use the discrete Grönwall inequality from Proposition~\ref{prop2} on the inequality derived in Lemma~\ref{lemma7}. We have: 
\begin{align}
    \ds_i'- \de_i - \ds_i & \leq \left( -1 - \frac{\mu}{\kappa} (\nu + 1 ) \right) \de_i + \frac{\mu}{\kappa} \left(\frac{\kappa}{\beta} - g(0) \right) \\ & + \left(2 + \frac{\mu \nu}{\kappa} \right) \left(\sum_{j=1}^{i-1} \ds_j' - \de_j - \ds_j\right)
\end{align}
Defining $p_i =   \left( -1 -\frac{\mu}{\kappa} (\nu + 1) \right) \de_i + \frac{\mu}{\kappa} \left(\frac{\kappa}{\beta} - g(0) \right)$, $q_i =\left( 2 + \frac{\mu\nu}{\kappa}  \right)$, and $f_{\ell} = 1$ for all $\ell$ in the Grönwall inequality, we have:
\begin{align*}
    \ds_i' - \de_i - \ds_i & \leq p_i + q_i \sum_{\ell = 1}^{i-1} p_{\ell} f_{\ell} \left(\prod_{k=\ell + 1}^{i-1} (1 + q_k f_k) \right) \\ 
    & = p_i + \left( 2 + \frac{\mu \nu}{\kappa} \right) \sum_{\ell = 1}^{i-1} p_{\ell} \prod_{k = \ell + 1}^{i-1} (1+q_k) \\ 
    & = p_i + \left( 2 + \frac{\mu \nu}{\kappa} \right) \sum_{\ell = 1}^{i-1} p_{\ell} \left(1+ 2 + \frac{\mu \nu}{\kappa} \right)^{i- \ell- 1} \\ 
    & = p_i + \left( 2 +\frac{\mu \nu}{\kappa} \right) \sum_{\ell = 1}^{i-1} p_{\ell} \left( 3 + \frac{\mu \nu}{\kappa} \right)^{i-\ell - 1}
\end{align*}
Adding $\de_i$ to both sides:
\begin{align*}
    \ds_i' - \ds_i \leq p_i + \de_i + \left( 2 +\frac{\mu \nu}{\kappa} \right) \sum_{\ell = 1}^{i-1} p_{\ell} \left( 3 + \frac{\mu \nu}{\kappa} \right)^{i-\ell - 1}
\end{align*}
and combining terms in $p_i$ with $\de_i$ we have: 
\begin{align*}
    \ds_i' - \ds_i \leq \left(-\frac{\mu}{\kappa} (\nu + 1) \right) \de_i + \frac{\mu}{\kappa} \left( \frac{\kappa}{\beta} - g(0) \right) + \left( 2 +\frac{\mu \nu}{\kappa} \right) \sum_{\ell = 1}^{i-1} p_{\ell} \left( 3 + \frac{\mu \nu}{\kappa} \right)^{i-\ell - 1}
\end{align*}
\end{proof}

\begin{lemma}
We can lower bound $\ds_i' - \de_i - \ds_i$ as: 
\begin{align}
    \ds_i' - \de_i - \ds_i & \geq  \left( - 1 +\frac{\kappa}{\mu}(\gamma + 1) \right) \de_i  + \frac{\kappa}{\mu}\left( \frac{\mu}{\beta} - g(0) \right) \\ & + \left( 2 + \frac{\kappa \gamma}{\mu} \right) \left( \sum_{j=1}^{i-1} \ds_j' - \de_j - \ds_j\right)
\end{align}
\end{lemma}
\begin{proof}
We use the upper bound on $\ds_i$ and lower bound on $\ds_i'$ derived in Lemmas~\ref{lemma1} and~\ref{lemma3} respectively to bound $\ds_i' - \ds_i$. We have: 
\begin{align*}
    \ds_i' - \ds_i & \geq -\left(\frac{\eta \mu}{\mu - \kappa} + \frac{\kappa}{\mu}(\gamma+1) - 1 \right) \de_i + \frac{\kappa}{\mu} \left(\frac{\mu}{\beta} - g(0) \right)  - \left(1 + \frac{\kappa \gamma}{\mu}\right) \tilde{u}_{i-1} \\ 
    & - \tilde{u}_{i-1}  - \left( 1 -\frac{\eta \mu}{\mu - \kappa} \right) \de_i \\ 
    & = \frac{\kappa}{\mu}(\gamma + 1) \de_i  + \frac{\kappa}{\mu}\left( \frac{\mu}{\beta} - g(0) \right) - \left( 2 + \frac{\kappa \gamma}{\mu} \right) \tilde{u}_{i-1}
\end{align*}
and adding $-\de_i$ to both sides gives us: 
\begin{align*}
    \ds_i' - \de_i - \ds_i &\geq \left( - 1 \frac{\kappa}{\mu}(\gamma + 1) \right) \de_i  + \frac{\kappa}{\mu}\left( \frac{\mu}{\beta} - g(0) \right) - \left( 2 + \frac{\kappa \gamma}{\mu} \right) \tilde{u}_{i-1} \\ 
    & = \left( - 1 +\frac{\kappa}{\mu}(\gamma + 1) \right) \de_i  + \frac{\kappa}{\mu}\left( \frac{\mu}{\beta} - g(0) \right) \\ 
    & + \left( 2 + \frac{\kappa \gamma}{\mu} \right) \left( \sum_{j=1}^{i-1} \ds_j' - \de_j - \ds_j \right)
\end{align*}
\end{proof}
\begin{lemma}\label{PNLLowerBound}
The sandwich profit $\mathsf{PNL}_i = \ds'_i - \ds_i$ can be lower bounded: 
\begin{align}
    \ds_i' - \ds_i \geq \de_i + \left( \frac{\mu}{\mu  + \kappa \gamma} \right)^i
\end{align}
\end{lemma}
\begin{proof}
We use the discrete Grönwall inequality from Proposition~\ref{prop3} on the inequality derived in Lemma~\ref{lemma8}. We have from Lemma \ref{lemma8}: 
\begin{align*}
    \ds_i' - \de_i - \ds_i & \geq  \left( - 1 +\frac{\kappa}{\mu}(\gamma + 1) \right) \de_i  + \frac{\kappa}{\mu}\left( \frac{\mu}{\beta} - g(0) \right) \\ 
    & + \left( 2 + \frac{\kappa \gamma}{\mu} \right) \left( \sum_{j=1}^{i-1} \ds_j' - \de_j - \ds_j\right)
\end{align*}
Negating this inequality, we have: 
\begin{align*}
   - \ds_i' +\de_i - \ds_i & \leq  \left(  1-\frac{\kappa}{\mu}(\gamma + 1) \right) \de_i  - \frac{\kappa}{\mu}\left( \frac{\mu}{\beta} - g(0) \right) \\ 
   & + \left( 2 + \frac{\kappa \gamma}{\mu} \right) \left( \sum_{j=1}^{i-1} - \ds_j' + \de_j + \ds_j\right)
\end{align*}
Defining $m_i = \left(  1-\frac{\kappa}{\mu}(\gamma + 1) \right) \de_i  - \frac{\kappa}{\mu}\left( \frac{\mu}{\beta} - g(0) \right)$, $q_i = \left( 2 + \frac{\kappa \gamma}{\mu} \right)$, and $f_{\ell} = 1$ for all $\ell$ in the Grönwall inequality, we have:
\begin{align*}
    - \ds_i' +\de_i - \ds_i & \leq m_i + q_i \sum_{\ell = 1}^{i-1} m_{\ell} f_{\ell} \left( \prod_{k =\ell+1}^{i-1} (1+ q_k f_k) \right) \\ 
    & = m_i + \left( 2 + \frac{\kappa \gamma}{\mu} \right) \sum_{\ell = 1}^{i-1} m_{\ell} \prod_{k=\ell+1}^{i-1} (1+q_k) \\ 
    & = m_i + \left( 2 + \frac{\kappa \gamma}{\mu} \right) \sum_{\ell = 1}^{i-1} m_{\ell} \left(1 +  2 + \frac{\kappa \gamma}{\mu} \right)^{i - \ell - 1} \\ 
    & = m_i + \left(2+ \frac{\kappa \gamma}{\mu} \right) \sum_{\ell = 1}^{i-1} m_{\ell} \left(3 + \frac{\kappa \gamma}{\mu} \right)^{i - \ell - 1}
\end{align*}
Once again negating, and adding $\de_i$ to both sides:
\begin{align*}
    \ds_i' - \ds_i \geq - m_i +\de_i - \left(2 + \frac{\kappa \gamma}{\mu} \right) \sum_{\ell = 1}^{i-1} m_{\ell} \left(3 + \frac{\kappa \gamma}{\mu} \right)^{i - \ell - 1}
\end{align*}
Combining $m_i$ with $\de_i$, we have: 
\begin{align*}
    \ds_i' - \ds_i \geq  \frac{\kappa}{\mu} ( \gamma +1 ) \de_i + \frac{\kappa}{\mu} \left( \frac{\mu}{\beta} - g(0) \right) - \left(2 + \frac{\kappa \gamma}{\mu} \right) \sum_{\ell = 1}^{i-1} m_{\ell} \left(3 + \frac{\kappa \gamma}{\mu} \right)^{i - \ell - 1}
\end{align*}
We note that $\frac{\kappa}{\mu}\left( \frac{\mu}{\beta} - g(0) \right) \geq 0$ whenever $\mu \geq g(0) \beta$ and $-1 + \frac{\kappa}{\mu}(\gamma +1) \geq 0 $ whenever $\gamma \geq \frac{\mu}{\kappa - 1}$, and imposing these conditions, we have: 
\begin{align*}
    \ds_i' - \de_i - \ds_i \geq \left( 2 + \frac{\kappa \gamma}{\mu}\right) \left( \sum_{j=1}^{i-1} \ds_j' - \de_j - \ds_j \right)
    \end{align*}
Applying Proposition~\ref{prop3}, and in particular noting that $q_r = -\left(2 + \frac{\kappa \gamma}{\mu} \right)$, $f_{\ell} = 1$ for all $\ell$, we have:
\begin{align*}
    \ds'_i - \de_i - \ds_i & \geq \prod_{\ell = 1}^{i} \left(1 - \left( 2 + \frac{\kappa \gamma}{\mu} \right)\right)^{-1} \\ & = \prod_{\ell = 1}^{i} \frac{-1}{1 + \frac{\kappa \gamma}{\mu}} \\ & = \left( \frac{-1}{1 + \frac{\kappa \gamma}{\mu}} \right)^{i}
\end{align*}
Adding $\de_i$ to both sides, we have:
\begin{align*}
    \ds_i' - \ds_i \geq \de_i + \left( \frac{-1}{1 + \frac{\kappa \gamma}{\mu}} \right)^{i}
\end{align*}
We take the bound for even $i$ to get: 
\begin{align*}
    \ds_i' - \ds_i \geq \de_i + \left( \frac{\mu}{\mu  + \kappa \gamma} \right)^i
\end{align*}
\end{proof}

\section{Proof of Proposition 1: Sandwich Pairwise Locality}\label{app:sandwichLocality}
We show conditions for sandwich attacks to be \textit{pairwise local}. That is, for any adjacent trades $\de_i,\de_{i+1}$:
\begin{align*}
    \mathsf{PNL}(\de_i + \de_{i+1}) \leq \mathsf{PNL}(\de_i) + \mathsf{PNL}(\de_{i+1})
\end{align*}
We use the bounds derived on $\mathsf{PNL}$ to show this. Recall that in the case of $\mathsf{PNL}(\de_i + \de_{i+1})$, the optimal sandwich for the composite trade $\de_i + \de_{i+1}$, $\ds_{i, i+1}$ satisfies:
\begin{align*}
    G\left( \de_i + \de_{i+1} + \ds_{i,i+1} + \sum_{j=1}^{i-1} \xi_j \right) - G\left( \ds_{i, i+1} + \de_{i+1} \sum_{j=1}^{i-1} \xi_j \right) = (1-\eta) G\left(\de_i + \de_{i+1} \right)
\end{align*}
Using Lemma \ref{PNLUpperBound}, we have: 
\begin{align*}
    \mathsf{PNL}(\de_i + \de_{i+1}) \leq \left(-\frac{\mu}{\kappa} (\nu + 1) \right) (\de_i + \de_{i+1}) + \frac{\mu}{\kappa} \left( \frac{\kappa}{\beta} - g(0) \right) + \left( 2 +\frac{\mu \nu}{\kappa} \right) \sum_{\ell = 1}^{i-1} p_{\ell} \left( 3 + \frac{\mu \nu}{\kappa} \right)^{i-\ell - 1}
\end{align*}
for $p_i =   \left( -1 -\frac{\mu}{\kappa} (\nu + 1) \right) (\de_i + \de_{i+1}) + \frac{\mu}{\kappa} \left(\frac{\kappa}{\beta} - g(0) \right)$.

Similarly, we have lower bounds for $\mathsf{PNL}(\de_i)$ and $\mathsf{PNL}(\de_{i+1})$ from Lemma \ref{PNLLowerBound}, which gives: 
\begin{align*}
    \mathsf{PNL}(\de_i) \geq \de_i + \left(\frac{\mu}{\mu + \kappa \gamma} \right)^i
\end{align*}
and
\begin{align*}
    \mathsf{PNL}(\de_{i+1}) \geq \de_{i+1} + \left(\frac{\mu}{\mu + \kappa \gamma} \right)^{i+1}
\end{align*}
Combining these bounds, we have: 
\begin{align*}
    & \mathsf{PNL}(\de_i + \de_{i+1}) - \mathsf{PNL}(\de_i) - \mathsf{PNL}(\de_{i+1}) \\ 
    & \leq \left(-\frac{\mu}{\kappa} (\nu + 1) \right) (\de_i + \de_{i+1}) + \frac{\mu}{\kappa} \left( \frac{\kappa}{\beta} - g(0) \right) + \left( 2 +\frac{\mu \nu}{\kappa} \right) \sum_{\ell = 1}^{i-1} p_{\ell} \left( 3 + \frac{\mu \nu}{\kappa} \right)^{i-\ell - 1} \\ 
    & - \de_i - \left(\frac{\mu}{\mu + \kappa \gamma} \right)^i - \de_{i+1} - \left(\frac{\mu}{\mu + \kappa \gamma} \right)^{i+1} \\ 
    & = \left( \frac{\mu}{\kappa}(\nu +1 ) -1\right) (\de_i + \de_{i+1}) +\frac{\mu}{\kappa} \left( \frac{\kappa}{\beta} - g(0) \right) \\ 
    & + \left( 2 +\frac{\mu \nu}{\kappa} \right) \sum_{\ell = 1}^{i-1} p_{\ell} \left( 3 + \frac{\mu \nu}{\kappa} \right)^{i-\ell - 1}  - \left(\frac{\mu}{\mu + \kappa \gamma} \right)^i - \left(\frac{\mu}{\mu + \kappa \gamma} \right)^{i+1}
\end{align*}
This bound gives us a sufficient condition for when $\mathsf{PNL}(\de_i + \de_{i+1}) - \mathsf{PNL}(\de_i) - \mathsf{PNL}(\de_{i+1}) \leq 0$. In particular, we need:
\begin{align}\label{eq: sufficientConditionLocality}
    &\left( \frac{\mu}{\kappa}(\nu +1 ) -1\right) (\de_i + \de_{i+1}) +\frac{\mu}{\kappa} \left( \frac{\kappa}{\beta} - g(0) \right) \nonumber\\ 
    & + \left( 2 +\frac{\mu \nu}{\kappa} \right) \sum_{\ell = 1}^{i-1} p_{\ell} \left( 3 + \frac{\mu \nu}{\kappa} \right)^{i-\ell - 1}  - \left(\frac{\mu}{\mu + \kappa \gamma} \right)^i - \left(\frac{\mu}{\mu + \kappa \gamma} \right)^{i+1} \leq 0 
\end{align}
for all $i \in [n]$. 
\section{Proof of Proposition 3}\label{app:prop3}
Recall by Lemma \ref{PNLUpperBound} that: 
\begin{align*}
    \ds_i' - \ds_i & \leq \left(-\frac{\mu}{\kappa} (\nu + 1) \right) \de_i + \frac{\mu}{\kappa} \left( \frac{\kappa}{\beta} - g(0) \right) + \left( 2 +\frac{\mu \nu}{\kappa} \right) \sum_{\ell = 1}^{i-1} p_{\ell} \left( 3 + \frac{\mu \nu}{\kappa} \right)^{i-\ell - 1}  
\end{align*}
Suppressing the constants, we write this as: 
\begin{align*}
    \mathsf{PNL}_{i} = \ds_i' - \ds_i \leq a \de_i + b + c \sum_{\ell = 1}^{i-1} d^{i - \ell - 1} \de_{\ell}
\end{align*}
Now, by Lemma \ref{PNLLowerBound} we have:
\begin{align*}
   \ds_i' - \ds_i \geq  \de_i + \left( \frac{\mu}{\mu  + \kappa \gamma} \right)^i
\end{align*}
and once again suppressing constants: 
\begin{align*}
    \mathsf{PNL}_i = \ds_i' - \ds_i \geq \de_i + e^i
\end{align*}
where $\ds_i' - \ds_i = \mathsf{PNL}_i$.

\section{Proof of Proposition 4}\label{app:prop4} 
Using the bound from Proposition \ref{prop:const_curvature} we have for a permutation $\pi$:
\begin{align*}
    \mathsf{PNL}_{\pi(i)} - \mathsf{PNL}_i & \leq a \de_{\pi(i)} + b + c \left(\sum_{\ell =1}^{\pi(i)-1} d^{\pi(i) - \ell - 1} \de_\ell \right)- \de_i - e^i
\end{align*}
and correspondingly, we have a lower bound: 
\begin{align*}
    \mathsf{PNL}_{\pi(i)} - \mathsf{PNL}_i & \geq \de_{\pi(i)} + e^{\pi(i)} - a \de_i - b - c \sum_{\ell = 1}^{i-1} d^{i - \ell - 1} \de_{\ell}
\end{align*}
Applying $\max$ and taking expectations over $\pi \sim S_k$: 
\begin{align*}
\Expect_{\pi \sim S_k} \left[ \max_{i \in [k]} \,\lvert \mathsf{PNL}_{\pi(i)} - \mathsf{PNL}_{i} \rvert \right] & \leq \Expect_{\pi \sim S_k} \left[ \max_{i\in [k]} a \de_{\pi(i)} + b + c \left(\sum_{\ell =1}^{\pi(i)-1} d^{\pi(i) - \ell - 1} \de_\ell\right) - \de_i - e^i \right]
\end{align*}
We now adapt the methodology used by \cite{chitra2021differential} to get our final bound.
First, define the partial sums $R_i(T_k, \pi) = a \de_{\pi(i)} + b + c \left(\sum_{\ell =1}^{\pi(i) - 1} d^{\pi(i) - \ell - 1} \Delta_{\ell} \right) - \de_i -e^i$ and consider the binary search tree $\mathsf{BST}(\mathbf{R}(T_k, \pi))$ whose root is $R_1(T_k, \pi)$.
The elements $R_j(T_k, \pi)$ are added sequentially to this tree.
Then, following \cite[\S3]{chitra2021differential}, we have the following bounds:
\begin{align*}
    \max_i | \mathsf{PNL}_{\pi(i)} - \mathsf{PNL}_i | & \leq |R_1 ( T_k , \pi) | + \max_j \left| R_j(T_k, \pi) \right| \mathsf{height}(\mathsf{BST}(\mathbf{R}(T_k, \pi))) 
\end{align*}
Now, recalling from \cite{reed2003height} that for equiprobable permutations $\Expect_{\pi \sim S_k} [ \mathsf{height}(\mathsf{BST}(\mathbf{R}(T_k, \pi)))] = \alpha \log k - \beta \log \log k$, we have: 
\begin{align*}
   \Expect_{\pi \sim S_k}\left[ \max_i | \mathsf{PNL}_{\pi(i)} - \mathsf{PNL}_i | \right] & \leq \Expect_{\pi \sim S_k} \left[ |R_1 (T_k , \pi)| + \max_j \left| R_j(T_k, \pi) \right| \mathsf{height}(\mathsf{BST}(\mathbf{R}(T_k, \pi))) \right] \\
   &=  \Expect_{\pi \sim S_k} |R_1 (T_k, \pi) |] + \left(\max_j |R_j(T_k, \pi)|\right) \Expect_{\pi\sim S_k}[\mathsf{height}(\mathsf{BST}(\mathbf{R}(T_k, \pi)))] \\ 
   &\leq \Expect_{\pi \sim S_k} |R_1 (T_k, \pi) |]  \\ 
   & +\max_{i,j} \left| a \de_{i} + b + c \sum_{\ell = 1}^{i - 1} d^{i - \ell -1} \de_{\ell} - \de_j - e^j\right| (\alpha \log k - \beta \log \log k) 
\end{align*}
where the last inequality uses the following identity 
\[
\max_j \left| a \de_{\pi(j)} + b + c \sum_{\ell =1}^{\pi(j) - 1} d^{\pi(j) - \ell - 1} \de_{\ell} - \de_j - e^j\right| \leq \max_{i,j} \left| a \de_{i} + b + c \sum_{\ell =1}^{i - 1} d^{i - \ell - 1} \de_{\ell} - \de_j - e^j\right|
\]
\noindent Now, note that we can bound:
\begin{align*}
    \Expect_{\pi \sim S_k} | R_1(T_k, \pi)| = \frac{1}{k} \sum_{j =1}^{k} \left|a \de_{j} + b + c \sum_{\ell =1}^{j - 1} d^{j - \ell -1} \de_{\ell} - \de_1 - e^1\right| \leq \max_{i,j}\left|a \de_{i} + b + c \sum_{\ell =1}^{i - 1} d^{i - \ell -1} \de_{\ell} - \de_j - e^j\right|
\end{align*}
which gives us the bound:
\begin{align*}
    \Expect_{\pi \sim S_k}\left[\max_i | \mathsf{PNL}_{\pi(i)} - \mathsf{PNL}_i |\right] \leq \max_{i,j}\left|a \de_{i} + b + c \sum_{\ell =1}^{i - 1} d^{i - \ell -1} \de_{\ell} - \de_j - e^j\right| (\alpha \log k - \beta \log \log k) 
\end{align*}
which allows us to conclude that $\Expect_{\pi \sim S_k}[\max_i | \mathsf{PNL}_{\pi(i)} - \mathsf{PNL}_i |] = O(\log k)$.

\section{Proof of Proposition 5}\label{app:prop5} 
Recall that we have the following lower bound on $\mathsf{PNL}_{\pi(i)} - \mathsf{PNL}_i$ from Section \ref{app:prop3}:
\begin{align*}
    \mathsf{PNL}_{\pi(i)} - \mathsf{PNL}_i & \geq \de_{\pi(i)} + e^{\pi(i)} - a \de_i - b - c \sum_{\ell = 1}^{i-1} d^{i - \ell - 1} \de_{\ell}
\end{align*}
Now, taking absolute values and averages, we have:
\begin{align*}
    \frac{1}{n} \sum_{i=1}^{n} \lvert \mathsf{PNL}_{\pi(i)} - \mathsf{PNL}_i \rvert & \geq \frac{1}{n} \sum_{i=1}^{n} \lvert  \de_{\pi(i)} + e^{\pi(i)} - a \de_i - b - c \sum_{\ell = 1}^{i-1} d^{i - \ell - 1} \de_{\ell} \rvert \\ 
    & \geq \frac{1}{n} n \min_i \left[ \de_{\pi(i)} + e^{\pi(i)} - a \de_i - b - c \sum_{\ell = 1}^{i-1} d^{i - \ell - 1} \de_{\ell}\right] \\ 
    & \geq \min_i \left[ \de_{\pi(i)} + e^{\pi(i)} - a \de_i - b - c \sum_{\ell = 1}^{i-1} d^{i - \ell - 1} \de_{\ell} \right] 
\end{align*}
which allows us to conclude that $\frac{1}{n} \sum_{i=1}^{n} \lvert \mathsf{PNL}_{\pi(i)} - \mathsf{PNL}_i \rvert = \Omega(1)$. 

\section{Proof of Theorem 1}\label{app:thm1}
We combine Propositions \ref{propBoundMax} and \ref{propBoundAvg} to get the main result:
\begin{align}\label{eq: costoffeudalism}
  \mathsf{CoF}(T_n) & = \frac{\Expect_{\pi \sim S_n} \left[ \max_{i \in [n]} \lvert \mathsf{PNL}_{\pi(i)} - \mathsf{PNL}_{i} \rvert \right]}{\Expect_{\pi \sim S_n} \left[ \frac{1}{n} \sum_{i=1}^{n} | \mathsf{PNL}_{\pi(i)} - \mathsf{PNL}_{i}|\right] } \nonumber \\ 
  & \leq \frac{\max_{i,j}\left|a \de_{i} + b + c \sum_{\ell =1}^{i - 1} d^{i - \ell -1} \de_{\ell} - \de_j - e^j\right| (\alpha \log n - \beta \log \log n)}{\Expect_{\pi \sim S_n} \left[ \frac{1}{n} \sum_{i=1}^{n} | \mathsf{PNL}_{\pi(i)} - \mathsf{PNL}_{i}|\right]} \nonumber \\ 
  & \leq \frac{\max_{i,j}\left|a \de_{i} + b + c \sum_{\ell =1}^{i - 1} d^{i - \ell -1} \de_{\ell} - \de_j - e^j\right| (\alpha \log n - \beta \log \log n)}{\min_i \left[ \de_{\pi(i)} + e^{\pi(i)} - a \de_i - b - c \sum_{\ell = 1}^{i-1} d^{i - \ell - 1} \de_{\ell} \right] } \nonumber \\ 
  & = O(\log n)
\end{align}
\section{Routing MEV}\label{app:routing}
In this section, we provid proofs of the optimality conditions for the CFMM Pigou example and Proposition \ref{prop6}. The latter establishes the locality of sandwich attacks on CFMM networks. 

\subsection{CFMM Pigou Example}\label{app: pigou}
We now derive the optimality condition for the optimal routing problem with no sandwiching: 
\[
\begin{aligned}
&\mathrm{maximize} && G_1(\Delta_1) + G_2(\Delta_2)\\
&\mathrm{subject\ to} && \Delta = \Delta_1 + \Delta_2\\
&&& \Delta_1, \Delta_2 \ge 0. 
\end{aligned}
\]
Setting up the Lagrangian, we have:
\begin{align*}
    \mathcal{L}(\Delta_1, \Delta_2, \nu, \mu_1, \mu_2) = G_1(\Delta_1) + G_2(\Delta_2) + \nu(\Delta - \Delta_1 - \Delta_2)  -\mu_1 \Delta_1 - \mu_2 \Delta_2
\end{align*}
Recall that any optimal $\Delta_1^*, \Delta_2^*$ must satisfy: 
\begin{align*}
    \frac{\partial \mathcal{L}}{\partial \Delta_1} =  \frac{\partial \mathcal{L}}{\partial \Delta_2} = 0 
\end{align*}
and the complementary slackness conditions: 
\begin{align*}
    \mu_1 \Delta_1^* = 0 \\
    \mu_2 \Delta_2^* = 0
\end{align*}
The first (stationarity) conditions give us: 
\begin{align*}
    G_1'(\Delta_1^*) = \nu + \mu_1 \\
    G_2'(\Delta_2^*) = \nu + \mu_2
\end{align*}
Therefore, for any $\Delta_1^*, \Delta_2^* > 0$, we must have: $G_1'(\Delta_1^*) = g_1(\Delta_1^*) = \nu^* = g_2(\Delta_2^*) = G_2'(\Delta_2^*)$, as desired. 

\subsection{Proof of Proposition \ref{prop6}}\label{app:prop6}
\paragraph{Defining Sandwich Profit on a Graph.}
In order to prove the bounds of Proposition \ref{prop6}, we first need to define what the sandwich profit will be.
This profit function can then be used to implicitly define sandwich sizes and we can use $(\mu,\kappa)$-smoothness to construct the bounds mentioned.


\begin{defn}
Define the \emph{cumulative output without sandwiches} as
\begin{align*}
    \bar{G}_{p_i}(\alpha_p \Delta) = G_{p_i^{k_i}} (G_{p_i^{k_{i} -1}} (\dots( G_{p_i^2}(G_{p_i^1}(\alpha_p \de)))))
\end{align*}
\end{defn}
The interpretation of this quantity is as the amount the user is expecting to receive from the path under no sandwiching. We can now write a defining equation for $\ds_e$ on every edge $e \in 1,\dots, |p|$ as a function of the flow entering that edge. We write this equation in words first and then incorporate symbols:
At the terminal node of path $p$, we know we need to receive $(1-\eta) \bar{G}_p(\de) = (1-\eta) G_T(\de)$ units of output token out. Now, look at the node immediately preceding it. Call this node $e$. We can write an equation: 
\begin{align*}
    G_e(\de_e + \ds_e) - G_e(\ds_e) = (1-\eta) G_T(\de) - G_e(\hat{G}_{e-1} (\de))
\end{align*} 
where $\hat{G}_{e-1}(\de)$ is the profit up to node $e-1$.
\paragraph{Implied Slippage Limits over a Path.}
We seek an explicit representation of $\ds_e$ in terms of the flow entering that edge, $\de_e$ and the \textit{global} slippage limit $\eta$.
Suppose that we have a path $\mathcal{P} = (e_1, \ldots, e_{T})$ where $e_{T}$ is the terminal edge (\eg~returns desired output token when traversed).
To do this, we write a Bellman-type equation that writes $\ds_{e_i}$ on every edge as a function of the terminal slippage $\eta$ and the slippages that occurred before.
For the final node, we have the following equivalence.
\begin{align*}
    G_{e_{T-1}}(\de_{e_{T-1}} + \ds_{e_{T-1}}) - G_{e_{T-1}}(\ds_{e_{T-1}}) = (1-\eta) G_T(\de_T) - G_{e_{T-2}}(\ds_{e_{T-2}})
\end{align*}
The left hand side of this equation represents the excess price impact that occurs at edge $e_i$ when the path $\mathcal{P}$ is traversed.
The right hand side is contribution to the terminal impact (a boundary term) from the $e_{i-1}$th edge.
This effectively says the flow into $e_{T-1}$ needs to be routed such that it exactly compensates for the excess price impact plus the output quantity.
Another way of framing this condition is as a divergence-free condition for the flow (\eg~input flow and output flows have to be equal in terms of their net price impact).
Similarly, we can recursively construct slippage limits for each $e_{i}$ as
\begin{align}\label{eq:recursionSandwich}
        G_{e_{i-1}}(\de_{e_{i-1}} + \ds_{e_{i-1}}) - G_{e_{i-1}}(\ds_{e_{i-1}}) = (1-\eta_i) G_{e_{i}}(\de_{e_i}) - G_{e_{i-2}}(\ds_{e_{i-2}})
\end{align}
From this we have a sequence of $T-1$ equations for solving for $T-1$ unknown variables $\eta_{e_1}, \ldots, \eta_{e_{T-1}}$.
This can be solved via dynamic programming, as this is an analogue of the Kolmogorov backward equation, albeit for slippage limits.
Therefore, there exists a unique way to solve for implied slippage limits along a route $\eta$.

This means that the net amount of output token the user receives from the CFMM network under sandwiching must be no more than $1-\eta$ times the amount the user would have received under no sandwiching.
The optimal sandwich attacks $\ds_e$ solve the above equation \eqref{eq: sandwichConditionRouting}. As there is just one equation for the network, but $|E|$ sandwiches to be solved for, we provide a heuristic that can be used to solve for each individual sandwich $\ds_e$ using a fixed point iteration, and use the solution that results to provide price of anarchy bounds for the network. 

\paragraph{Proof of \eqref{eq:ds-edge-bd} and \eqref{eq:d-edge-bd}.}\label{app:thm2}
We use the equations \eqref{eq:recursionSandwich} to construct the bounds on $\ds_e$ described in Proposition \ref{prop6}.
We assume that we have uniform upper and lower bounds on all $G_{e}(\cdot)$. That is, $\kappa \Delta \leq G_e (\Delta) \leq \mu \Delta$ for all $e \in E$. Recall the equations:
\begin{align*}
    G_{T-1} (\de_{T-1} + \ds_{T-1}) - G_{e_{T-1}}(\ds_{e_{T-1}}) = (1-\eta) G_{T-1}(G_{T-2} (\dots (\Delta))) 
\end{align*}
for the terminal sandwich $\ds_{T-1}$ and:
\begin{align*}
    G_{i} (\de_i + \ds_i) - G_{i}(\ds_{i}) = (1-\eta) G_{T-1}(G_{T-2} (\dots (\Delta)))  \\- G_{i}( G_{i-1}( G_{i-2}(\dots (G_{1}(\de) + \ds_1) + \dots \ds_{i-2})+\ds_{i-1}))
\end{align*}
for the intermediate sandwiches $\ds_i$ for $i = 1, \dots, T-2$.
Let $\de_{T-1} = G_{T-2}( G_{T-3}(\dots (G_{1}(\de) + \ds_1) + \dots \ds_{T-3})+\ds_{T-2})$.
We now apply $\mu$ and $\kappa$ bounds to the above equation to get: 
\begin{align*}
    \de_{T-1} \leq \mu^{T-2}(\de + \ds_1) + \mu^{T-3} \ds_{2} + \dots + \mu \ds_{T-2}
\end{align*}
Which gives us:
\begin{align*}
    G_{T-1}(\de_{T-1} + \ds_{T-1} ) - G_{T-1} (\ds_{T-1}) \leq \mu^{T-1} (\de + \ds_1) + \mu^{T-2} \ds_2 
    \\ + \dots + \mu^2 \ds_{T-2} + \mu \ds_{T-1} - \kappa \ds_{T-1} 
\end{align*}
and 
\begin{align*}
    & (1-\eta) G_T(G_{T-1}(G_{T-2} (\dots (\Delta)))) \\ & - G_{T-1}( G_{T-2}( G_{T-3}(\dots (G_{1}(\de) + \ds_1) + \dots \ds_{T-3})+\ds_{T-2}))  \\  & \leq (1-\eta) \mu^{T-1} \de - \kappa^{T-1}(\de + \ds_1) + \kappa^{T-2} \ds_2
 + \dots + \kappa \ds_{T-2}
\end{align*}
Forcing the bound on the RHS to be greater than the bound on the LHS, we have: 
\begin{align*}
    \mu^{T-1} (\de + \ds_1) + \mu^{T-2} \ds_2 
     + \dots + \mu^2 \ds_{T-2} + \mu \ds_{T-1} - \kappa \ds_{T-1} \\ \geq (1-\eta) \mu^{T-1} \de - \kappa^{T-1}(\de + \ds_1) + \kappa^{T-2} \ds_2
 + \dots + \kappa \ds_{T-2}
\end{align*}
Moving all the $\ds_{T-j}$ terms for $j > 1$ to the RHS, we have: 
\begin{align*}
    (\mu - \kappa) \ds_{T-1} &\leq (\mu^{T-1} - \kappa^{T-1}) \Delta - (\eta +1) \mu^{T-1} \Delta \\
    &- (\mu^{T-1} + \kappa^{T-1}) \ds_1 -\dots - (\mu^2 + \kappa^2) \ds_{T-2}
\end{align*}
and dividing by $\mu - \kappa$:
\begin{align}\label{eq:ds-t}
     \ds_{T-1} &\leq \frac{1}{\mu - \kappa} \left( (\mu^{T-1} - \kappa^{T-1}) \Delta - (\eta +1) \mu^{T-1} \Delta \right. \nonumber \\
     &\left. - (\mu^{T-1} + \kappa^{T-1}) \ds_1 -\dots - (\mu^2 + \kappa^2) \ds_{T-2} \right)
\end{align}
Recall the defining recursion for a sandwich:
\begin{align}\label{eq:mu-kappa}
    G_{e_{i-1}} (\de_{e_{i-1}} + \ds_{e_{i-1}} ) - G_{e_{i-1}} (\ds_{e_{i-1}}) = (1-\eta_i) G_{e_i}(\de_{e_i}) - G_{e_{i-2}} (\ds_{e_{i-2}})
\end{align}
Using the $(\mu, \kappa)$-smoothness of $G_{e_{i-1}}$ we can upper bound the left hand side and lower bound the right hand side (which is the defining relation for $\eta_i$) as: 
\begin{align*}
    \mu (\de_{e_{i-1}} + \ds_{e_{i-1}} ) - \kappa \ds_{e_{i-1}} \leq (1-\eta_i) \kappa \de_{e_i} - \mu \ds_{e_{i-2}} 
\end{align*}
Rearranging and collecting terms: 
\begin{align*}
    \ds_{e_{i-1}} \leq - \frac{\mu}{\mu - \kappa} \ds_{e_{i-2}} + \left( \frac{ (1-\eta_i)\kappa - \mu }{\mu-\kappa}\right) \de_{e_i} \leq -\de_{e_i} - \frac{\mu}{\mu - \kappa} \ds_{e_{i-2}} 
\end{align*}
If combining this equation with \eqref{eq:ds-t} gives:
\begin{align}\label{eq:ds-t}
     \ds_{T-1} &\leq \frac{1}{\mu - \kappa} \left( (\mu^{T-1} - \kappa^{T-1}) \Delta - (\eta +1) \mu^{T-1} \Delta \right. \nonumber \\
     &\left. + (\mu^{T-1} + \kappa^{T-1}) \ds_{e_1} -\dots - (\mu^2 + \kappa^2) \left(\de_{T-2} + \frac{\mu}{\mu-\kappa} \ds_{T-2}\right) \right)
\end{align}
Note that this gives an upper bound on the terminal sandwich, which implies a bound on the total path sandwich attack $\ds_p$. Solving the recursions using Propositions \ref{prop2} and \ref{prop3} for this bound yields Eq. \eqref{eq:ds-edge-bd}.
We can compute a similar bound for $\Delta_{e_i}$ using the other bound for \eqref{eq:mu-kappa} and arrive at a similar bounded recursion, yielding \eqref{eq:d-edge-bd}.

\end{document}